\documentclass[aps,pra,reprint,a4paper,superscriptaddress,floatfix,longbibliography]{revtex4-2}
\usepackage[british]{babel}
\usepackage{graphicx}
\usepackage{adjustbox}
\usepackage{epsfig}
\usepackage{microtype}
\usepackage{amsfonts}
\usepackage{amsmath}
\usepackage{amssymb}
\usepackage{amsthm}
\usepackage{enumitem}
\usepackage{footmisc}
\usepackage{color}
\usepackage{tikz}
\usepackage{scalerel}
\usepackage[colorlinks=true,linkcolor=blue,citecolor=magenta,urlcolor=blue]{hyperref}
\usepackage{cleveref}
\usepackage{comment}
\usepackage{braket}
\usepackage{physics}
\usepackage{bm}
\usepackage{xfrac}
\usepackage{soul}

\theoremstyle{plain}
\newtheorem{theorem}{Theorem}
\newtheorem{proposition}[theorem]{Proposition}
\newtheorem{corollary}[theorem]{Corollary}

\theoremstyle{definition}

\theoremstyle{remark}
\newtheorem{Remark}[theorem]{Remark}
\theoremstyle{plain}

\newcommand{\ie}{i.e.\ }
\newcommand{\eg}{e.g.\ }
\newcommand{\Vorobev}{Vorob{\textquotesingle}ev}
\newcommand{\THb}{TH}

\newcommand{\possproblem}[1]{{\color{orange}[POSS PROBLEM]#1}}

\usepackage{soul}

\newcommand{\del}[1]{{\color{gray}\st{#1}}}
\newcommand{\hili}[1]{{\color{blue}#1}}

\usepackage{todonotes}
\newcommand{\rsbs}[1]{\todo[color=magenta]{#1}}
\newcommand{\rsbi}[1]{\todo[color=magenta,inline]{#1}}
\newcommand{\rsbshili}[1]{\todo[color=blue]{#1}}
\newcommand{\rsbihili}[1]{\todo[color=blue,inline]{#1}}

\newcommand{\rci}[1]{\todo[color=red,inline]{#1}}

\begin{document}

\title{Exclusivity principle, Ramsey theory, and $n$-cycle PR boxes}


\newcommand{\inllong}{INL -- International Iberian Nanotechnology Laboratory, Av.~Mestre Jos\'e Veiga s/n, 4715-330 Braga, Portugal}
\newcommand{\inlshort}{INL -- International Iberian Nanotechnology Laboratory, Braga, Portugal}
\newcommand{\haslablong}{HASLab, INESC TEC, Universidade do Minho, Departamento de Informática, Campus de Gualtar, 4710-057 Braga, Portugal}
\newcommand{\haslabshort}{HASLab, INESC TEC, Universidade do Minho, Braga, Portugal}
\newcommand{\diumlong}{Departamento de Inform\'atica, Universidade do Minho, Campus de Gualtar, 4710-057 Braga, Portugal}
\newcommand{\diumshort}{Departamento de Inform\'atica, Universidade do Minho, Braga, Portugal}

\author{Raman Choudhary}
\affiliation{\inlshort}
\affiliation{\haslabshort}
\affiliation{\diumshort}

\author{Rui Soares Barbosa}
\affiliation{\inlshort}


\begin{abstract}
The exclusivity principle (E-principle) states that the sum of probabilities of pairwise exclusive events cannot exceed 1. 
Since quantum theory satisfies this principle, any correlation that violates it is necessarily post-quantum.
Unlike other principles proposed to characterize quantum correlations, its intrinsically non-bipartite formulation enables its application in more general contextuality scenarios, where it can single out correlations among the set of non-disturbing ones.
Although equivalent to the no-signalling condition for any bipartite Bell scenario, this equivalence breaks down for two independent copies of the same scenario.
Such violation of the E-principle due to multiple copies, known as its \emph{activation effect}, was studied in [Nat Commun \textbf{4}, 2263 (2013)] for the nonlocal extremal boxes of $(2,m,2)$, $(2,2,d)$, and $(3,2,2)$ Bell scenarios.
The authors mapped the problem of exhibiting activation effects to finding certain cliques inside joint exclusivity graphs.
In this work, we refine the joint exclusivity structure to be an edge-colored exclusivity multigraph.
This allows us to draw a novel connection to Ramsey theory, which guarantees the existence of certain monochromatic subgraphs in sufficiently large edge-colored cliques, providing a powerful tool for ruling out E-principle violations.
We then exploit this connection, drawing on Ramsey-theoretic results to study violations of the E-principle by multiple copies of the contextual extremal boxes of $n$-cycle scenarios, called $n$-cycle PR boxes. For the usual ($n=4$) PR box we show that the known E-principle violation of $\sfrac{5}{4}$ is the maximal achievable using two copies, and that this same upper bound applies to two copies of the KCBS ($n = 5$) PR box.
We then prove that $n \geq 6$-cycle PR boxes, unlike the extremal boxes of the aforementioned Bell scenarios, do not exhibit activation effects with two or three copies.
Finally, for any number of independent copies $k$, we establish a lower bound on $n$ above which $n$-cycle PR boxes do not exhibit activation effects with $k$ copies.
\end{abstract}


\maketitle


\tableofcontents

\section{Introduction}
\label{sec:introduction}
\subsection{Context}
\label{ssec:context}
Quantum theory, unlike the theory of special relativity, lacks a description where simple physical principles guide its mathematical machinery.
Identifying such principles is one of the major goals in quantum foundations research,
pursued under programmes aiming at `reconstructions of quantum theory'.
Such efforts seek to enhance our understanding of quantum theory and diminish the mystery around its interpretations. 

Several principles have been proposed to characterize quantum correlations, such as
information causality, macroscopic locality, non-trivial communication complexity, or no advantage for non-local computation.
However, Gallego et al.~\cite{gallego2011quantum} challenged the effectiveness of these principles in singling out quantum correlations by providing examples of post-quantum correlations in tripartite Bell scenarios satisfying each of them. This revealed a key shortcoming in their bipartite formulations, highlighting the need for multipartite versions.

Subsequently, Fritz et al.~\cite{fritz2013local} introduced a new intrinsically multipartite principle, the Local Orthogonality principle, and explored its ability to detect post-quantum correlations in certain Bell scenarios. This principle also applies to Kochen--Specker (KS) contextuality scenarios more broadly, where it is commonly known as the exclusivity principle (or E-principle).
It states that the sum of probabilities of a set of pairwise exclusive events \footnote{An event here refers to the joint outcome of a set of compatible measurements. Two events are deemed exclusive when they assign different outcomes to a common measurement, \eg the events $[{+}{+}|AB]$ and $[{-}{-}|AC]$ differ on the common measurement $A$.}  must be less than or equal to $1$. 

The E-principle was used to witness post-quantumness of extremal correlations in the $(2,2,d)$, $(2,m,2)$, and $(3,2,2)$ Bell scenarios \cite{fritz2013local}.
The method made use of the \emph{activation effect} of the E-principle,
where a correlation may independently satisfy the E-principle but the joint consideration of multiple independent copies of that correlation `activates' a violation.
The simplest example is that of Popescu--Rohrlich (PR) boxes: while a single PR box satisfies the E-principle, joint events coming from two independent PR boxes violate it.
Since quantum correlations satisfy the E-principle and thus so do independent products of quantum correlations, this violation witnesses the post-quantumness of the PR box itself.
That is, the fact that two PR boxes violate the E-principle certifies that a single PR box is not quantum realizable.
Similarly, it turns out that for any non-local extremal correlation of the $(2,2,d)$, $(2,m,2)$, and $(3,2,2)$ Bell scenarios, two copies are enough to detect post-quantumness via the E-principle \cite{fritz2013local}.

A systematic approach to study such activation effects uses graph theory.
An exclusivity graph is built whose nodes represent possible events of a correlation and whose edges represent exclusivity between such events.
The activation setup comprising multiple independent correlations is then captured by constructing the joint exclusivity graph as the OR product of the individual exclusivity graphs.
A violation of the E-principle manifests as the existence of certain cliques inside this joint exclusivity graph.
Ref.~\cite{fritz2013local} exploited computational tools to find such cliques, but while that method worked well for relatively small graphs, in general the problem of finding maximal cliques is known to be NP-hard \cite{karp1975computational}.

Besides certifying post-quantumness of the aforementioned Bell non-local correlations, the E-principle has provided other interesting insights, which we briefly review.
Cabello~\cite{cabello2013simple} used the E-principle to derive the quantum maximum of the famous KCBS non-contextuality inequality. Cabello et al.~\cite{cabello2013basic} showed that the exclusivity graphs of non-contextuality inequalities violated by quantum theory necessarily contain certain induced sub-graphs, and further provided numerical evidence that the quantum maxima of the non-contextuality inequalities corresponding to those basic exclusivity graphs can be singled out by the E-principle. Amaral et al.~\cite{amaral2014exclusivity} showed that for self-complementary exclusivity graphs, the E-principle singles out the quantum set of correlations, known as the theta body of the graph
\footnote{Note that the theta body of an exclusivity graph contains all quantum correlations consistent with that graph and not just within a given scenario. That is, it contains quantum correlations coming from all possible scenarios that produce events realizing the given exclusivity graph.}.
This result was later generalized by Cabello~\cite{cabello2019quantum} to  arbitrary exclusivity graphs.
Furthermore, it is known that any theory that follows Specker's principle immediately satisfies the E-principle, including quantum and classical theories \cite{specker1990logik}. But the converse does not hold: the set of almost quantum correlations is inconsistent with Specker's principle~\cite{gonda2018almost} yet satisfies the E-principle~\cite{navascues2015almost}.  

\subsection{Contributions}
\label{ssec:contributions}
This work extends the study of the E-principle for detecting post-quantumness of correlations to the broader domain of KS contextuality scenarios, advancing the analyses from Ref.~\cite{fritz2013local} which were restricted to Bell scenarios.

Our first contribution is to establish a hitherto undiscovered connection with Ramsey theory \cite{graham1991ramsey,li2022elementary}, an (in)famous field within combinatorics. 
This provides a powerful theoretical framework for analyzing activation of the E-principle  without relying on computational methods.
We then apply this general link to study E-principle activation effects for the extremal contextual correlations of the fundamental class of $n$-cycle scenarios \cite{araujo2013all}. These extremal correlations are known as $n$-cycle PR boxes because of their structural similarity with the usual PR boxes.
The exploitation of known Ramsey theory results allows us to generate a plethora of results about the activation effects of $n$-cycle PR boxes without any further computational help.
This is especially relevant in cases where the size of violation-producing cliques shoots up exponentially (in the number of independent copies), hence rendering brute-force searches completely impractical.
We now provide some more details about both aspects of our contribution.

We set up the connection with Ramsey theory by introducing a refinement to the conventional analysis of activation effects, replacing the joint exclusivity graph by an edge-colored multigraph.
A unique color is assigned to the exclusivity graph of each independent correlation, allowing one to preserve the information about which individual copies are responsible for determining exclusivity between two joint events. By contrast, in the conventional approach, this detailed information is erased by the OR product.
The problem of finding violations of the E-principle then becomes that of finding edge-colored cliques inside this multigraph. Such edge-colored cliques are precisely the objects of study of Ramsey theory when applied to graphs.

Ramsey theory investigates an intriguing phenomenon about how order emerges from randomness: as combinatorial structures grow larger, certain patterns become inevitable, as expressed by the popular catchphrase `infinite chaos is impossible'.
When applied to graph theory, it studies the minimal size an edge-colored complete graph must be to force the appearance of specific monochromatic substructures regardless of the particular coloring of the edges.
A classic illustration involves a complete graph whose edges can be colored either red or green: once the graph contains at least six vertices, it must contain either a red triangle or a green triangle.

Applying results from Ramsey theory yields several significant results for $n$-cycle PR boxes:
\begin{itemize}
    \item  First, we show that two copies are enough to observe violation of the E-principle only for the cases $n=4$ and $n=5$, \ie for the CHSH and KCBS scenarios.
This identifies $n$-cycle PR boxes with $n \geq 6$ as the first known correlations that do not violate the E-principle using only two copies.
\item We then exploit one of the few results from the Ramsey theory literature that considers more than two colors to show that three copies are still insufficient to produce E-principle violations for $n \geq 6$-cycle scenarios.
\item 
Using another known Ramsey-theoretic result, we find a lower bound on the cycle size $n$ as a function of the number of copies $k$ such that the $k$ independent copies of an $n$-cycle PR box are not enough to violate the E-principle.
\item  We also use the canonical Ramsey theory example highlighted above to show that for the $4$-cycle and $5$-cycle PR boxes, the previously known violation of the E-principle is in fact the maximum violation that can be produced with two copies. This implies that, in these $n=4$ and $n=5$ cases, two copies can only violate the E-principle by the minimum conceivable amount.
\end{itemize}

\subsection{Structure of the paper}
\label{ssec:structure_of_paper}
In \Cref{sec:preliminaries} we introduce the mathematical background required to study activation effects of the E-principle.
\Cref{ssec:scenarios_and_correlations} introduces contextuality scenarios and correlations, and \Cref{ssec:egraphs} explains how an exclusivity graph can be drawn from a given correlation. \Cref{ssec:E_principle} then introduces the E-principle, while \Cref{ssec:joint_E_graph} explains how the problem of finding activations of the E-principle maps to finding cliques inside a joint exclusivity graph.

In \Cref{sec:refinement} we explain our refinement of the conventional method to study activation effects. 
\Cref{ssec:multi_color_or_product} explains how the joint exclusivity structure becomes an edge-colored multigraph.
\Cref{ssec:monochromatic_projection} defines its monochromatic projections, which are constrained by Ramsey theory results as explained in \Cref{ssec:Ramsey}.

In \Cref{sec:cycle_scenarios} we present the specific scenarios on which our results focus.
\Cref{ssec:n-cycle_scenario} introduces $n$-cycle scenarios, and then \Cref{ssec:n-cycle_PRboxes} presents their extremal contextual correlations, the $n$-cycle PR boxes.
\Cref{ssec:violation} gives the condition for violation of the E-principle by $k$ of independent copies of $n$-cycle PR boxes.

In \Cref{sec:result_highlights} we highlight the results of our analysis of activation effects of the E-principle for $n$-cycle PR boxes.
Sections~\ref{ssec:k_1}--\ref{ssec:k_4} respectively cover the cases where one, two, three, or more copies are considered.

In \Cref{sec:Proofs} we provide the proofs of our main results.
\Cref{ssec:E-graphs_and_properties} shows how the exclusivity graphs of $n$-cycle PR boxes look like and what relevant subgraphs they contain.
As a warm-up, \Cref{ssec:impossible} proves that a trivial clique exists within the joint exclusivity graph of $k$ copies of any $n$-cycle PR box only one node short of a violation, but that it cannot be extended to a violation-producing clique.
Then, \Cref{ssec:activation_eff,ssec:exploiting_ramsey_theory} contain the main results about activation effects for the $n$-cycle PR boxes, with
\Cref{ssec:exploiting_ramsey_theory} covering those that leverage known Ramsey-theoretic results.

Finally, in \Cref{sec:discussion}, we conclude our presentation. \Cref{ssec:conclusion} recaps the main contributions, focusing on the methodological aspects. \Cref{ssec:outlook} suggests and examine some possible directions for future follow-up research. 
\section{Preliminaries}
\label{sec:preliminaries}
This section covers the necessary background, introducing the E-principle and the methodology to study its activation effects.
\subsection{Measurement scenarios and correlations}
\label{ssec:scenarios_and_correlations}
Measurement scenarios abstract the notion of an experimental setup where certain measurements can be performed on a system but not all of them can be performed simultaneously.
A \emph{measurement scenario} is defined by specifying a set of measurements, a set of possible outcomes for each measurement, and a compatibility relation among the measurements.
A set of pairwise compatible measurements constitutes a context, representing measurements that can be jointly performed
\footnote{Note that we consider only compatibility structures determined by a binary relation or `compatibility graph': a set of measurements is compatible if it is pairwise so. This corresponds to Specker's principle, which in particular holds for quantum mechanics.}.
In quantum theory, measurements are represented as PVMs and compatibility corresponds to commutativity of the corresponding operators.
A more formal treatment and further details can be found \eg in \cite[Section II]{choudhary2024lifting}.

For a given measurement scenario, particular outcome statistics are specified by a \emph{correlation}:
a probability distribution on the joint measurement outcomes for each context.

\Cref{tab:correlation_table} shows
a correlation on the $(2,2,2)$ Bell scenario, also known as the CHSH scenario.
This is a Bell nonlocality scenario where two parties (Alice and Bob)
each have two alternative measurement setting ($A_0$ and $A_1$ for Alice, $B_0$ and $B_1$ for Bob),
each with two possible outcomes (${+}$ and ${-}$).
As a contextuality scenario, capturing the constraint that each party can perform only one of its measurements,
two distinct measurements are compatible if and only if they belong to different parties.
So, this is a four-measurement scenario with four maximal contexts: $\{A_i,B_j\}$ for $i,j \in \{0,1\}$.
A correlation is specified by a probability distribution on joint outcomes for each of these four contexts,
corresponding to the rows of the table.
The specific correlation shown in \Cref{tab:correlation_table} is known as the Popescu--Rohrlich (PR) box \cite{popescu1994quantum},
a well-studied super-quantum correlation that maximally violates the CHSH inequality.
It is an extremal correlation of the $(2,2,2)$ Bell scenario, \ie a vertex of its polytope of no-signalling correlations.

\begin{table}[h]
    \centering
    \begin{adjustbox}{max width=\textwidth}
    \begin{tabular}{|c|c|c|c|c|}
        \hline
         & $++$ & $+-$ & $-+$ & $--$ \\
         \hline
        $A_0B_0$ & $\sfrac{1}{2}$ & 0 & 0 & $\sfrac{1}{2}$ \\ \hline
        $A_0B_1$ & 0 & $\sfrac{1}{2}$ & $\sfrac{1}{2}$ & 0 \\ \hline
        $A_1B_0$ & $\sfrac{1}{2}$ & 0 & 0 & $\sfrac{1}{2}$ \\ \hline
        $A_1B_1$ & $\sfrac{1}{2}$ & 0 & 0 & $\sfrac{1}{2}$ \\ \hline
    \end{tabular}
    \end{adjustbox}
    \caption{Correlation table for the PR box}
    \label{tab:correlation_table}
\end{table}

\subsection{Exclusivity graphs}
\label{ssec:egraphs}

Given a correlation, one constructs its \emph{exclusivity graph} of possible events.
The nodes of this graph represent joint outcome events over maximal contexts that occur with non-zero probability,
\eg the event $[{+}{+} | A_0B_0]$ for the PR box.
Two vertices are connected by an edge whenever the two events are \emph{exclusive},
meaning that there is a common measurement to which they assign different outcomes;
for example, the events $[{+}{+} | A_0B_0]$ and $[{-}{+} | A_0B_1]$ are exclusive
since they differ on the common measurement $A_0$.
\Cref{fig:E_PR} depicts the exclusivity graph of the PR box.
One may think of exclusivity graphs of correlations as being vertex-weighted: the weight of a node is the probability assigned by the correlation to the event it represents.
In the case of the PR box, all nodes are assigned the same weight, $\sfrac{1}{2}$.
Such uniform weighting is typical of the class of exclusivity graphs studied in this text,
namely those arising from extremal correlations in cycle scenarios ($n$-cycle PR boxes) and their products; see \Cref{sec:cycle_scenarios}.
Consequently, we mostly focus on the structure of the underlying unweighted graphs.

\subsection{The exclusivity principle}
\label{ssec:E_principle}

The \emph{exclusivity principle} (or \emph{E-principle}) states that the sum of probabilities of a set of mutually exclusive events must be less than or equal to $1$.

To test this principle on a given correlation, it suffices to consider its (vertex-weighted) exclusivity graph.
Since edges in this graph represent exclusivity, its cliques (fully connected subsets of nodes) correspond to sets of mutually exclusive events, each determining an E-principle inequality.
This renders cliques crucial structures for studying the exclusivity principle using graph theory.
Verifying whether a correlation satisfies the exclusivity principle then amounts to checking the inequalities corresponding to all the maximal cliques of its exclusivity graph.
This quickly becomes intractable, since it requires enumerating all the maximal cliques of a graph,
which is NP-hard 
\footnote{More precisely, the problem of deciding whether a graph has a clique of a given size is NP-complete \cite{karp1975computational}. Listing all maximal cliques clearly allows one to solve this decision problem. Moreover, since there are graphs with exponentially many maximal cliques, in the worst case their enumeration requires exponential time.}.
Still, various software packages are available which can be used to find cliques within modestly sized graphs.


\begin{figure}
    \centering
    \includegraphics[width=0.6\linewidth]{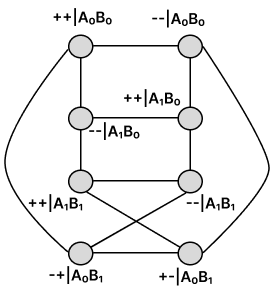}
    \caption{Exclusivity graph of the PR box}
    \label{fig:E_PR}
\end{figure}

Violations of the E-principle exhibit what are known as \emph{activation effects}, where correlations that individually satisfy the E-principle may violate it when combined.
In particular, there exist correlations that satisfy the E-principle on their own, yet lead to a violation when multiple independent copies are considered together.
The PR box provides a concrete example: while a single PR box satisfies the E-principle,
two independent PR boxes (\ie the correlation obtained as the product of two PR boxes) violate it.

In graph-theoretic terms, the violation is witnessed by a clique in the joint exclusivity graph of the two (or more) independent correlations, even though the individual exclusivity graphs admit no violation-witnessing cliques.

\subsection{Joint exclusivity graph of independent boxes}
\label{ssec:joint_E_graph}
To study activation effects, we must consider the product of correlations that captures their independent parallel composition.
Here the word `independent' evokes an analogy with independent random variables in probability theory, which are jointly distributed according to the product measure.
The situation we have in mind is that of the two (or more) measurement scenarios being run in parallel and producing outcomes independently of each other.

To keep the analysis at the level of exclusivity graphs, we construct the joint exclusivity graph from the exclusivity graphs of the individual correlations.
This involves specifying how independent events combine into joint events, and defining the exclusivity relation between joint events from the individual ones, as well as the probabilistic weights assigned to joint events.
In short, a joint event is represented by a tuple comprising one event drawn from each individual graph, and two joint events are deemed exclusive if their component events are exclusive in at least one of the individual graphs.
The probabilistic independence of the correlations is reflected in that the probability of a joint event is the product of the probabilities of its constituent events.

This idea is formally captured through the \emph{OR product}, also known as the \emph{co-normal product}, of graphs.
Let $G_i$ be the exclusivity graph of the $i$th correlation for $i \in \{1,\ldots,k\}$.
The joint exclusivity graph of these $k$  independent correlations will be denoted $J^k(G_1,\ldots,G_k)$, or $J^k$ for short leaving the graphs implicit. It is defined as follows.
The vertex set of $J^k$ is 
\begin{equation}
    V(J^k) =  V(G_1) \times \cdots \times V(G_k),
\end{equation}
where `$\times$' denotes the Cartesian product.
That is, a node of $J^k$ is a $k$-tuple $u=(u_1, \ldots ,u_k)$ where $u_i \in V(G_i)$, \ie each $u_i$ is a node in $G_i$.
The edge set of $J^k$ is defined as
\begin{equation}
    E(J^k)= \{(u,v) \mid   \exists\,  i \in \{1, \ldots, k\}.\; (u_i,v_i) \in E(G_i)\}.
\end{equation}
Finally, we define the vertex-weighting on $J^k$: owing to the independence of the correlations,
the probability $P(u)$ of a joint event $u = (u_1,\ldots,u_k)$ is given by the product
\begin{equation}
P(u) = P(u_1) \cdots P(u_k).
\label{Eqn:2}
\end{equation}

Note that the OR product of graphs does not retain information about which individual graph(s) gives rise to each edge in the product graph.
As a result, $J^k$ fails to encode which of the $k$ original exclusivity graphs determine exclusivity between a given pair of joint events.

\section{Refining the joint exclusivity graph}
\label{sec:refinement}
We introduce a fine-grained version of the joint exclusivity graph, using colors to track the source of exclusivity between pairs of joint events.
We then show how each monochromatic subgraph inside this new structure relates to the corresponding exclusivity graph in $\{G_1, \ldots,G_k\}$, and we explain how Ramsey theory constrains these monochromatic subgraphs, and so the original exclusivity graphs as well.

\subsection{Joint multigraph of independent boxes}
\label{ssec:multi_color_or_product}
We improve upon the construction of the joint exclusivity graph $J^k$ to encode which individual graphs contribute to exclusivity between joint events.
To this end, we introduce a refined construction, the \emph{multicolor product} of graphs, producing a joint exclusivity structure $\Gamma^k$ that is an edge-colored multigraph.
The key idea is that we assign a unique color $i$ to the edges arising from each original graph $G_i$.
A pair of nodes may thus be connected by multiple edges of different colors, with each color appearing at most once and indicating that the two joint events are exclusive due to the corresponding component.
\Cref{fig:graph_product} illustrates this construction with a toy example.

Formally, given exclusivity graphs $G_1, \ldots, G_k$, the joint exclusivity multigraph is given by their multicolor product
$\Gamma^k = \Gamma^k(G_1,\ldots,G_k)$, defined as follows.
The vertex set $V(\Gamma^k)$ is exactly the same as that of $J^k$:
\begin{equation}
    V(\Gamma^k) = V(G_1) \times \cdots \times V(G_k).
\end{equation}
The crucial difference lies in the edge structure.
Between any two nodes $u,v \in V(\Gamma^k)$, representing joint events, there is a separate edge for each individual graph $G_i$ which has an edge between the corresponding components $u_i$ and $v_i$.
Hence, the edge set $E(\Gamma^k)$ is partitioned into color classes:
 \begin{equation}
    E(\Gamma^k) = \bigcup_{i=1}^{k} E_{i}(\Gamma^k) 
\end{equation}
where each $E_{i}(\Gamma^k)$ is the set of edges with color $i$, which are derived from the $i$th individual graph $G_i$:
\begin{equation}
E_{i}(\Gamma^k) = \{(u,v)_{i} \mid (u_i,v_i) \in E(G_i)\}.
\end{equation}
Here, $(u,v)_i$ denotes an edge of color $i$ between vertices $u$ and $v$.
The relationship to the usual joint exclusivity graph $J^k$ is given by
\begin{equation}
    E(J^k) = \{(u,v) \mid \exists i \in \{1, \ldots, k\}.\, (u,v)_i \in E(\Gamma^k)\}.
\end{equation}
Of course, the vertex-weighting indicating probability of joint events remains as for $J^k$, given by \Cref{Eqn:2}.

\begin{figure}
    \centering
    \includegraphics[width=1\linewidth]{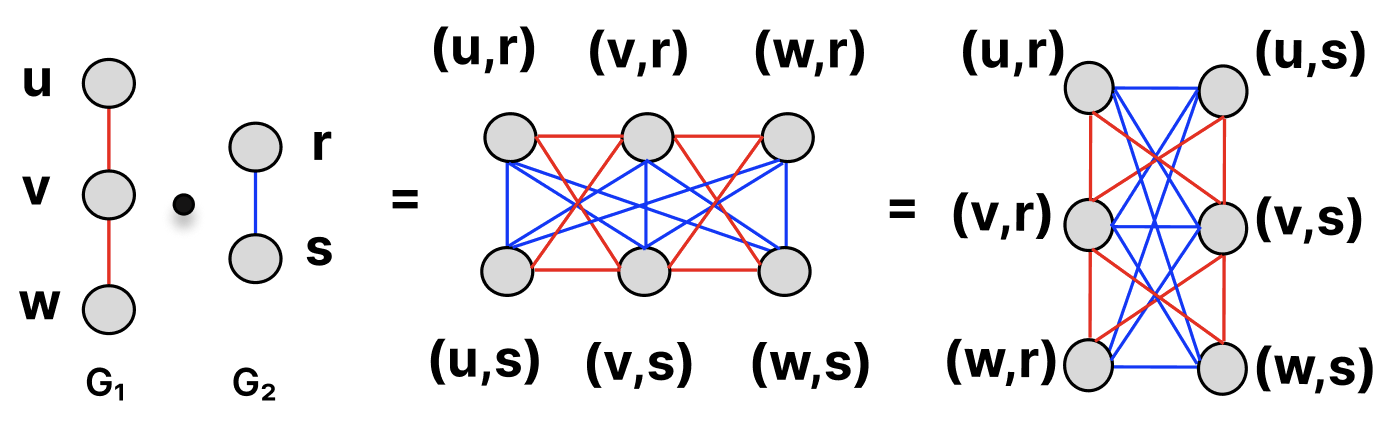}
    \caption{An example illustrating multicolor product of two graphs $G_1$ and $G_2$. The multigraph obtained can be thought to have multiple layers of nodes where each layer is isomorphic to $G_1$ (resp.\ $G_2$) with one of the colors, while the graph obtained by picking one node from each layer and the edges of the other color is necessarily isomorphic to $G_2$ (resp.\ $G_1$)  
    as shown in the middle (resp.\ rightmost) representation above.}
    \label{fig:graph_product}
\end{figure}

To test the E-principle on $\Gamma^k$, one seeks sets of nodes such that every pair of nodes in the set is connected by at least one colored edge.
Such a set of nodes forms an edge-colored clique in $\Gamma_k$, with the possibility of more than one edge (of different colors) between each pair of nodes. 
The corresponding events are pairwise exclusive, hence form a clique in $J^k$.
Thus, as before, a violation of the E-principle occurs when the sum of probabilities of the events in such an (edge-colored) clique exceeds $1$.

We remark that edge-colored multigraphs have been previously employed in a similar setting but with a different intention in Refs.~\cite{rabelo2014multigraph,vandre2022quantum,porto2024quantum},
where nodes represent events contributing to a given Bell inequality and edges denote exclusivity, with the color indicating the party from which such exclusivity arises.

\subsection{Monochromatic projection of $\Gamma^k$}
\label{ssec:monochromatic_projection}
\begin{figure}
    \centering
    \includegraphics[width=1\linewidth]{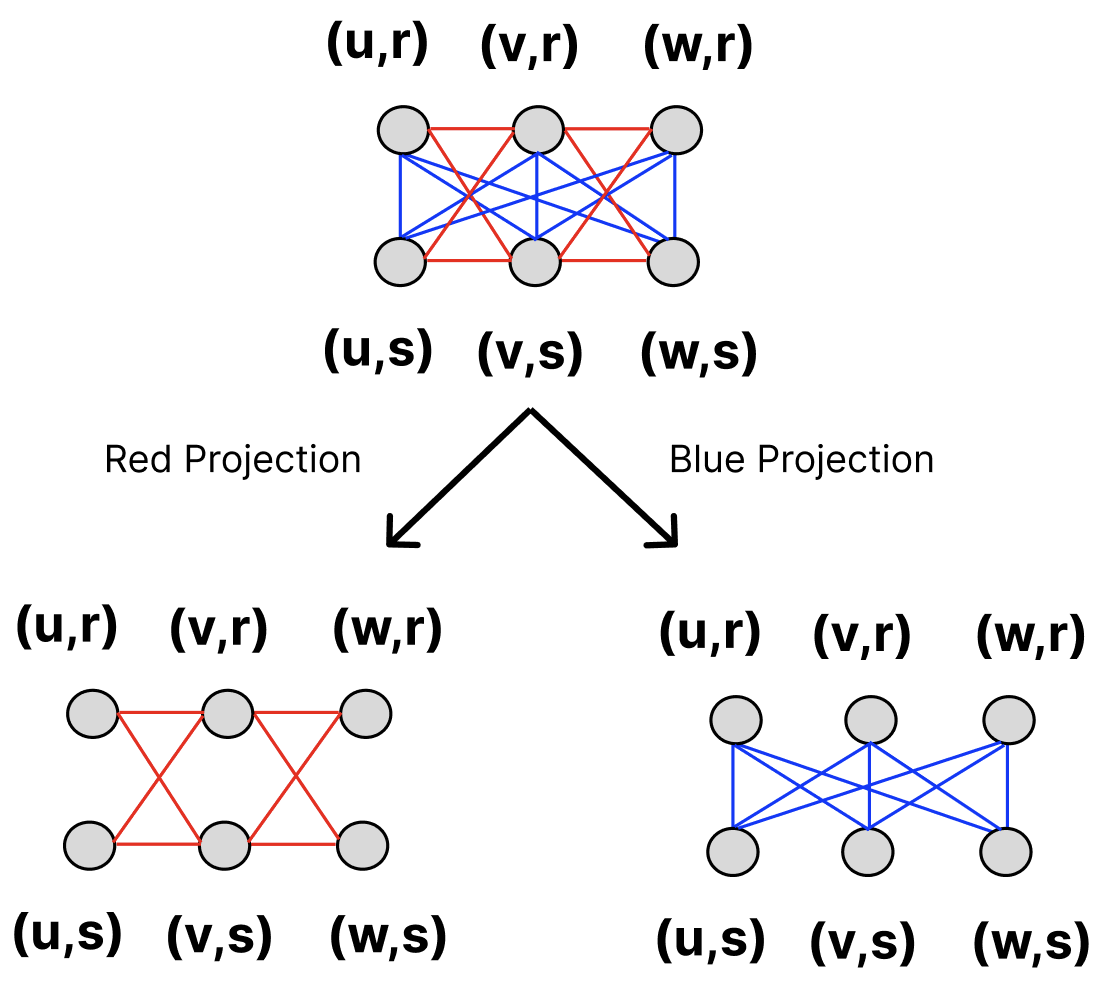}
    \caption{The joint exclusivity multigraph constructed in \Cref{fig:graph_product} and its monochromatic projections.}
    \label{fig:monochromatic_projection}
\end{figure}

In \Cref{ssec:Ramsey}, we will see how Ramsey theory necessitates the existence of certain monochromatic subgraphs inside sufficiently large cliques of $\Gamma^k$.
To facilitate that analysis and to draw meaningful conclusions from it, it is crucial to formally define the monochromatic projections of $\Gamma^k$ and to study the connections between such projections and the individual exclusivity graphs $G_i$.

For each color $i \in \{1, \ldots, k\}$,
the \emph{monochromatic projection} of color $i$ is defined as $H_i(\Gamma^k) = (V(\Gamma^k),E_{i}(\Gamma^k))$,
\ie $H_i(\Gamma^k)$ is the subgraph of $\Gamma^k$ with the same set of nodes but containing only the edges of color $i$.
Thus, it tracks only the exclusivity relations between joint events that are derived from the exclusivity graph $G_i$ of the $i$th correlation.
\Cref{fig:monochromatic_projection} illustrates the two monochromatic projections of the toy example multigraph from \Cref{fig:graph_product}.

To follow the proofs of our results,
it is imperative to highlight the precise relationship between the monochromatic projections $H_{i}(\Gamma^k)$ and the original exclusivity graphs $G_{i}$.

The first observation is that one can partition the nodes of $\Gamma^k$, and hence of $H_i(\Gamma^k)$,
according to their $i$th component, \ie by the equivalence relation whereby $u, u' \in V(\Gamma^k)$ are deemed equivalent if and only if $u_i = u'_i$.
For each node $v \in V(G_i)$, write
\begin{equation}
    V_v^i = \{u \in V(\Gamma^k) \mid u_i = v \}
\end{equation}
for the set of joint events (nodes of $\Gamma^k$) whose $i$th component is the fixed event $v$.
In other words, 
$V_v^i = V(G_1) \times \cdots \times V(G_{i-1}) \times \{v\} \times V(G_{i+1}) \times \cdots \times V(G_k)$.
These sets partition $V(\Gamma^k)$:
\begin{equation}
    V(\Gamma^k) = \bigcup_{v \in V(G_i)} V_v^i.
\end{equation}
Moreover, all the nodes in each of these sets are neighborhood-equivalent in $H_i(\Gamma^k)$, the monochromatic projection of color $i$:
if $(v,w) \in E(G_i)$ then every node in the class $V_v^i$ is connected to every node in the class $V_w^i$ by an edge of color $i$ in $\Gamma^k$, while if $(v,w) \notin E(G_i)$ then no nodes of $V_v^i$ are connected to any nodes of $V_w^i$ by an edge of color $i$ in $\Gamma^k$.
In the example from \Cref{fig:monochromatic_projection},
the equivalence classes with respect to the red graph $G_1$
are
$V_u^1 = \{(u,r), (u,s)\}$, 
$V_v^1 = \{(v,r),(v,s)\}$, 
and
$V_w^1 = \{(w,r),(w,s)\}$,
while the equivalence classes with respect to the blue graph $G_2$ are 
$V_r^2 = \{(u,r),(v,r),(w,r)\}$
and
$V_s^2 = \{(u,s),(v,s),(w,s)\}$.

The following proposition is an immediate consequence of the structure just observed.
\begin{proposition}
\label{Prop1}
    The monochromatic projection $H_i(\Gamma^k)$ contains a subgraph isomorphic to $G_i$.
\end{proposition}
\begin{proof}
    For each $v \in V(G_i)$, pick a node $\tilde{v} \in V(\Gamma^k)$ from the corresponding equivalence class $V_v^i$, \ie such that $\tilde{v}_i = v$.
    The induced subgraph of $H_i(\Gamma^k)$ determined by these nodes is isomorphic to $G_i$ since
    $(\tilde{v},\tilde{w})_i \in E(\Gamma^k)$ if and only if $(v,w) \in E(G_i)$.
\end{proof}


The central observation that underlies most of our results is expressed by the following proposition relating monochromatic odd cycles in $\Gamma^k$ with odd cycles in the corresponding original graph.

\begin{proposition}
    \label{Prop2}
    For each odd cycle in $H_i(\Gamma^k)$, there is an odd cycle in $G_i$ of at most the same length.
\end{proposition}
\begin{proof}
Let $C$ be an odd cycle in $H_i(\Gamma^k)$.
If all the nodes in $C$ belong to different equivalence classes, then their $i$th components all differ and thus form a cycle in $G_i$.
Otherwise, let $v$ and $v'$ be two nodes in the same equivalence class, \ie such that $v_i=v'_i$.
We show that, in that case, there exists a strictly smaller odd cycle in $H_i(\Gamma^k)$.

\Cref{fig:even_odd_cycle}(c) illustrates the basic idea of the proof.
Since nodes in the same equivalence class cannot be adjacent in $H_i(\Gamma^k)$, $v$ and $v'$ cannot be consecutive nodes in the cycle $C$.
Let $u$ and $w$ be the two neighboring nodes of $v$ in the cycle $C$. Being in the same equivalence class as $v$, the node $v'$ must also be connected to both $u$ and $w$.
Consider two paths starting from $v'$ and going in opposite directions along the cycle $C$: the path from $v'$ to $u$ and the path $v'$ to $w$, both avoiding the node $v$ (see \Cref{fig:even_odd_cycle}(a)).
Since $C$ is an odd cycle, then one of these paths must necessarily have even length. Adjoining to it the edge $(u,v')$ or the edge $(w,v')$, we obtain an odd cycle.
This cycle has strictly smaller length than $C$, for it is formed by a proper subset of its nodes having dropped at least the node $v$.

This process can be repeated until one reaches an odd cycle each of whose nodes belongs to a different equivalence class.
The fact that the length strictly decreases at each step guarantees termination, with the limit case being a cycle of length $3$, where each node must belong to a different equivalence class.
\end{proof}

\begin{figure}
    \centering
    \includegraphics[width=1\linewidth]{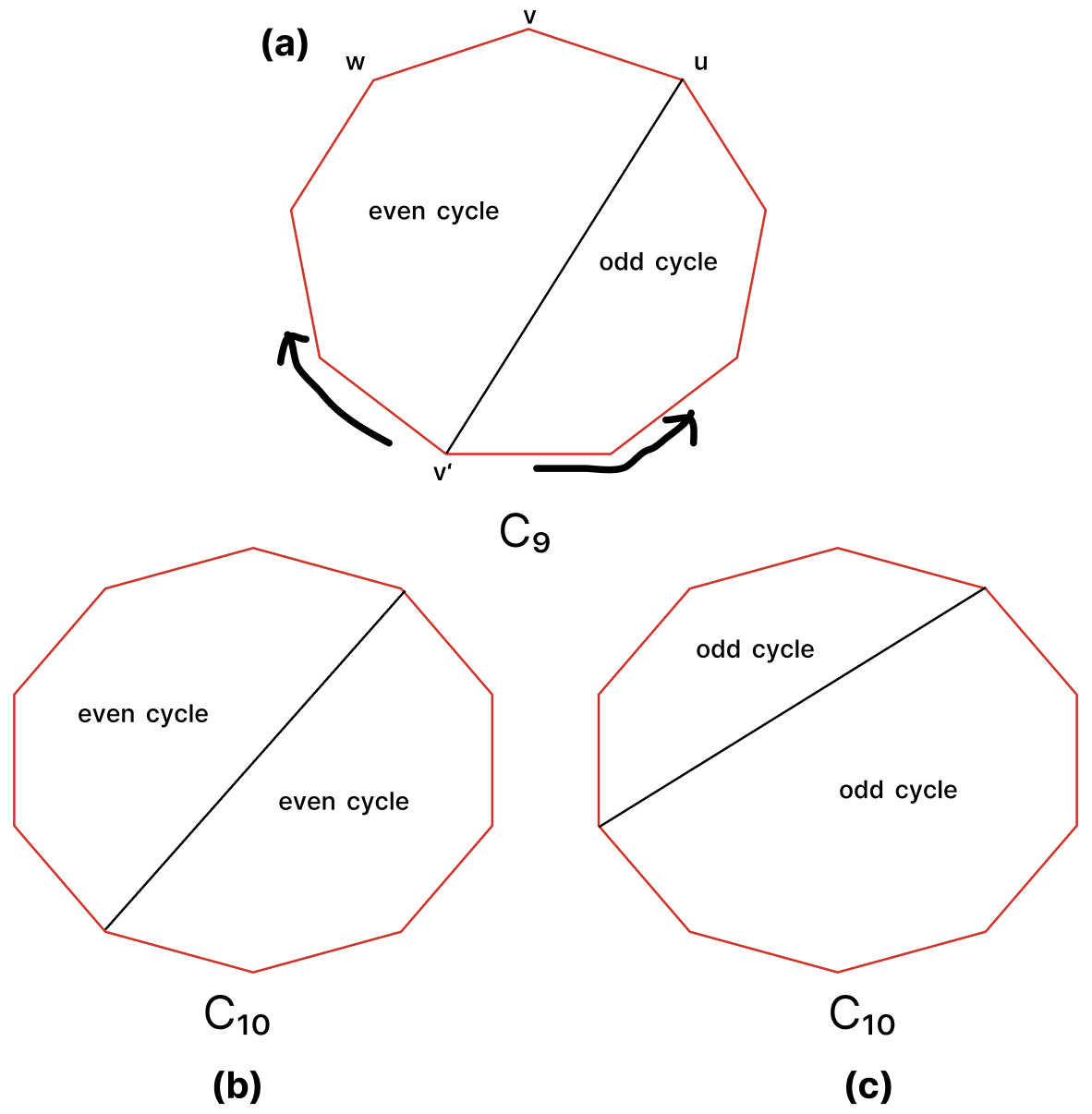}
    \caption{An edge between non-adjacent nodes in a cycle divides it into two smaller cycles. If the original cycle is odd, a chord necessarily leads to (a) an even and an odd cycle. If the original cycle is even, a chord can divide the cycle into (b) two even cycles or (c) two odd cycles. }
    \label{fig:even_odd_cycle}
\end{figure}

Note that a similar proof strategy fails for even cycles
because drawing an edge between two nodes in an even cycle leads to both partitions having even length or both having odd length, as shown respectively in \Cref{fig:even_odd_cycle}(b) and \Cref{fig:even_odd_cycle}(c).
By contrast, in the case of an odd cycle, one partition is guaranteed to have odd length, as shown in \Cref{fig:even_odd_cycle}(a), guaranteeing that the iterative process eventually reaches an odd cycle each of whose nodes belongs to a unique equivalence class, which in turn implies the existence of a cycle of the same length in $G_i$.
\begin{corollary}
\label{corollary:smallestoddcycle}
    \label{cor:smallest_odd_cycle_monochrome}
    The length of the smallest odd cycle in $H_i(\Gamma^k)$ is equal to the length of the smallest odd cycle in $G_i$.
\end{corollary} 
\begin{proof}
By \Cref{Prop2}, the length of the smallest odd cycle in $G_i$ is at most that of the smallest odd cycle of $H_i(G)$.
On the other hand, by \Cref{Prop1},  $H_i(\Gamma^k)$ contains a cycle that is isomorphic to the smallest odd cycle of $G_i$.
\end{proof}

Notice that in establishing the above relationships between $H_i(\Gamma^k)$ and $G_{i}$ we made no assumptions about the exclusivity graphs $G_1, \ldots, G_k$. Thus, the results hold in full generality, for any correlation in any measurement scenario.
Later, we make use of this fully general connection to explore the specific case of $n$-cycle PR boxes.
We now explain how violations of the E-principle (or lack thereof) for a joint exclusivity graph can be studied through the lens of Ramsey theory.

\subsection{Ramsey-theoretic constraints}
\label{ssec:Ramsey}
Ramsey theory is a branch of combinatorics that studies the emergence of regularity in arbitrary partitions of (large enough) combinatorial structures.
A typical question considers a structure that is cut into a set number of pieces and asks how large the structure must be to guarantee that, regardless of the partition, at least one of the pieces contains a certain pattern.

The prototypical Ramsey-theoretic problem problem -- which, along with minor variations, underlies our results --
can be formulated in the language of graph theory.
It considers colorings of the edges of a complete graph and asks for the minimal size that guarantees the presence of specific monochromatic substructures.
More precisely, given $k$ graphs $S_1, \ldots, S_k$, where each $S_i$ is associated with a color $i \in \{1, \ldots, k\}$, the Ramsey number $R(S_1, \ldots, S_k)$ denotes the minimum size of a complete graph such that any $k$ edge coloring contains a monochromatic $S_i$ in color $i$ for some $i$.
A simple example is $R(K_3,K_3)=6$: as already mentioned, any way of coloring the edges of $K_6$ must produce a monochromatic triangle, but the same is not true for $K_5$.
A more involved example is $R(K_3,K_3,K_4) = 30$ \cite{codish2016computing}: if three colors are available, then $K_{30}$ is the smallest complete graph such that any edge coloring contains a $K_3$ of color $1$ or $2$ or a $K_4$ of color $3$.
Reading it the other way round, once we know a certain Ramsey number $R(S_1,\ldots,S_k)$, we can conclude that for any $k$ edge-colored clique of size at least $R(S_1,\ldots,S_k)$ there is at least one color $i \in \{1,\ldots,k\}$ whose monochromatic projection contains $S_i$.

Since a multigraph allows for more edges than a graph with the same nodes, Ramsey-theoretic constraints $(S_1,\ldots,S_k)$ that apply to any $k$ edge coloring of a complete graph also apply to a $k$ edge coloring of a multigraph of the same size. Therefore, for our analysis, we do not necessarily need new Ramsey-theoretic results specific for multigraphs in order to squeeze out conclusions applicable to multigraphs like $\Gamma^k$.
In fact, the Ramsey theory of multigraphs is not sufficiently developed in the literature.
Our work is therefore based only on Ramsey-theoretic results established for ordinary graphs.

In our analysis, this theoretical framework becomes powerful because it constrains the possible structures within joint exclusivity graphs, allowing one to rule out violations of the exclusivity principle.
As established earlier in \Cref{ssec:multi_color_or_product}, a violation of the E-principle corresponds to certain cliques inside the edge-colored multigraph $\Gamma^k$.
Ramsey-theoretic results force monochromatic substructures to be present inside the cliques being sought for violation.
That is, the existence of violation-producing cliques within $\Gamma^k$ implies certain constraints on the monochromatic  projections $H_i(\Gamma^k)$.
Guaranteeing the absence of the required substructures within the corresponding projections
immediately precludes the existence of any violation-producing clique in $\Gamma^k$, ruling out the possibility of violating the E-principle.
This forms the logical basis of our results and also establishes a general connection between Ramsey theory and activation effects of the E-principle.
Concretely, in our proofs, we use a few Ramsey-theoretic results that necessitate the presence of odd cycles in the monochromatic projections $H_i(\Gamma^k)$. \Cref{Prop2} then implies a direct constraint on the original exclusivity graph $G_i$, which we can explicitly rule out.

Notably, while this approach is powerful for proving impossibility results, establishing the existence of E-principle violations through Ramsey theory is more challenging.
To prove impossibility, it suffices to show inconsistency between the monochromatic projections and the sought-after clique with just one known Ramsey-theoretic result. On the other hand, merely being consistent with Ramsey-theoretic constraints does not in itself guarantee the existence of a violation.

\section{Cycle scenarios}
\label{sec:cycle_scenarios}
We now introduce $n$-cycle scenarios and their contextual extremal boxes, termed $n$-cycle PR boxes.
Our results later explore activation effects for these correlations, mostly in light of Ramsey theory.  

\subsection{The $n$-cycle scenarios}
\label{ssec:n-cycle_scenario}
\begin{figure}
    \centering
    \includegraphics[width=1\linewidth]{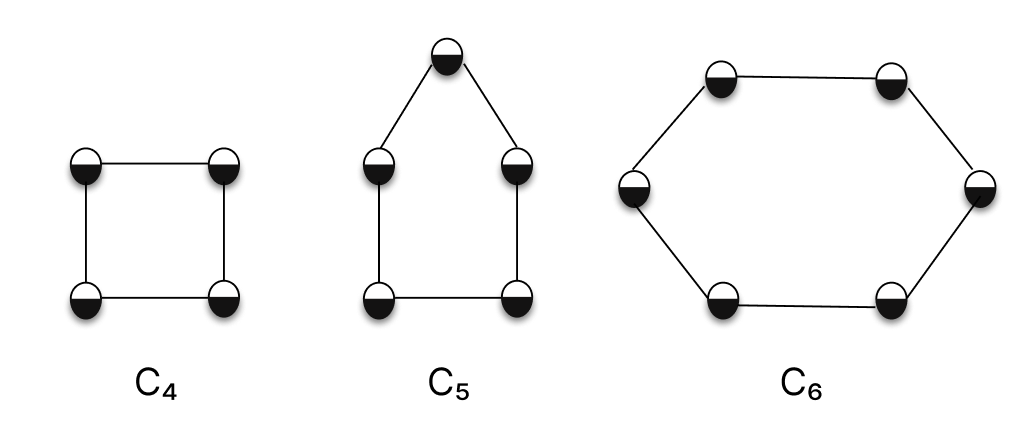}
    \caption{Compatibility graphs of the $n$-cycle scenarios for $n = 4,5,6$. Each bi-colored node represents a dichotomic measurement and each edge denotes joint measurability of connected measurements. $C_4$ corresponds to the CHSH scenario while $C_5$ is the KCBS scenario.}
    \label{fig:n-cycle}
\end{figure}

A cycle scenario is a scenario whose measurement compatibility structure is described by a cycle graph.
Cycle scenarios are arguably the most fundamental Kochen--Specker contextuality scenarios.
\Vorobev’s theorem \cite{vorob1962consistent} guarantees that any contextuality-witnessing measurement scenario must contain an $n$-cycle (for some $n\geq 4$) as an induced sub-scenario.

We focus on cycle scenarios in which all measurements are dichotomic, i.e. have two possible outcomes, which we label $+$ and $-$.
The $n$-cycle scenario has $n$ measurements $A_0, \ldots, A_{n-1}$,
with only consecutive measurements (modulo $n$) being compatible, so that the compatibility graph is the $n$-cycle graph $C_n$, see \Cref{fig:n-cycle}.
Hence, the maximal contexts are the two-measurement sets of the form $\{A_i,A_{i \oplus 1}\}$, where $\oplus$ denotes addition modulo $n$.
The simplest non-trivial (\ie contextuality-witnessing) example is that of the CHSH scenario, which corresponds to the case $n=4$.

Thanks to the simple structure of these (dichotomic) $n$-cycle scenarios, their sets of classical, quantum, and non-disturbing correlations are well understood.
Ara\'ujo et al.~\cite{araujo2013all} characterized their no-disturbance polytopes by providing all their extremal correlations (vertices),
characterized their classical (\ie noncontextual) polytopes by providing all their facet-defining noncontextuality inequalities,
and it further supplemented these inequalities with their corresponding maximal quantum violations
\footnote{Note that the classical polytope of a scenario is naturally given in V-representation: its vertices are the deterministic noncontextual models, corresponding to global assignments of outcomes to all the measurements. What is hard is thus to find its facets, \ie the noncontextuality inequalities. By contrast, the polytope of non-disturbing correlations is naturally given in H-representation: it is defined by affine equations expressing no-disturbance and normalization together with positivity inequalities. 
What is typically hard is thus to find its vertices, \ie the extremal correlations, particularly the nonclassical ones (the vertices of the noncontextual polytope are also vertices of the no-disturbance polytope).}.

\subsection{The $n$-cycle PR boxes}
\label{ssec:n-cycle_PRboxes}
For the purpose of this work, we are interested in the non-classical vertices of the no-disturbance polytopes of $n$-cycle scenarios.  These turn out to be a natural generalization of the PR box correlation from the CHSH scenario ($n = 4$), and we therefore refer to them as $n$-cycle PR boxes.
The $n$-cycle PR boxes are described as follows
(recall that a correlation consists of a probability distribution over the joint outcomes for each context):
\begin{enumerate}

    \item For each context $\{A_i,A_{i\oplus1}\}$, the joint outcomes $(a_i,a_{i\oplus1})$ in the support of the respective probability distribution satisfy either $a_i \cdot a_{i \oplus 1 } = 1$  or $a_i \cdot a_{i \oplus 1 } = -1$, where $a_i$ denotes the outcome of measurement $A_i$.
    In other words, for any given context, the only permissible joint outcomes are either the correlated pairs $\{++,--\}$ or the anti-correlated pairs $\{+-,-+\}$.

    \item For each context, the probability distribution is uniform: each of the two permissible joint outcomes occurs with the same probability, $\sfrac{1}{2}$
    \footnote{Such uniform probabilities are actually forced upon us  by no-disturbance once the support of the distributions is fixed (by the other two items). The fact that the possibilistic support fixes the values of the probabilities uniquely is actually a feature of any extremal correlation in any scenario \cite{abramsky2016possibilities}.}.
    
    \item The choices between correlation and anti-correlation for each context obey a global constraint:
    the number of contexts exhibiting anti-correlated outcomes ($a_i \cdot a_{i\oplus 1}  = -1$) is odd.
    Consequently,
    when $n$ is even, an odd number of contexts exhibit correlated outcomes, 
    while when $n$ is odd, an even number of contexts exhibit correlated outcomes.
\end{enumerate}

\begin{figure}
    \centering
    \includegraphics[width=1\linewidth]{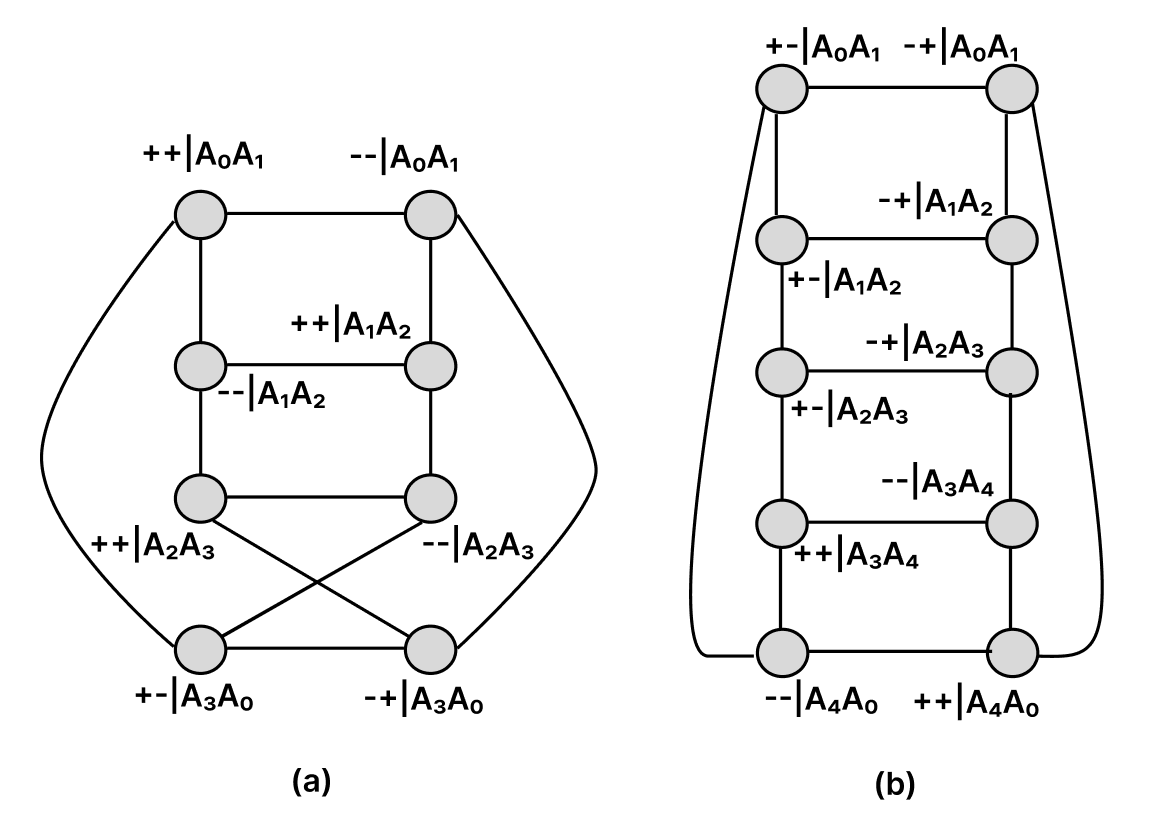}
    \caption{Simplest exclusivity graphs of: (a) n = 4-cycle PR box, the graph is isomorphic to the Möbius ladder graph $M_8$ (b) n = 5-cycle PR box, the graph is isomorphic to the Prism graph $Y_5$. An arbitrary $n$-cycle PR box is a generalized form of (a) or (b) above depending on whether $n$ is even or odd respectively.}
    \label{fig:4_5_cycle_PR_box}
\end{figure}

For a fixed $n$, all $n$-cycle PR boxes are equivalent under permutation of outcomes of individual measurements, and hence share the same exclusivity graph.
As we show in \Cref{theorem1} ahead, these exclusivity graphs follow two distinct patterns depending on the parity of $n$:
\begin{itemize}
    \item for even $n$, the exclusivity graph is the M\"obius ladder graph $M_{2n}$ shown in \Cref{fig:E-graph}(a), generalizing the $n=4$ CHSH case (\Cref{fig:4_5_cycle_PR_box}(a)).

    \item for odd $n$, the exclusivity graph is the circular ladder graph, also known as the Prism graph, $Y_n$ shown in \Cref{fig:E-graph}(b), generalizing the $n=5$ KCBS case (\Cref{fig:4_5_cycle_PR_box}(b))
\end{itemize}

\subsection{Violation condition}
\label{ssec:violation}

In the case of (products of) $n$-cycle PR boxes, the problem of E-principle violation admits an especially simple description.
This is due to the fact that in any $n$-cycle PR box each possible event occurs with the same probability, $\sfrac{1}{2}$.
Consequently, when dealing with $k$ independent copies of such boxes, by \cref{Eqn:2} each event in the joint exclusivity graph occurs with probability $\sfrac{1}{2^k}$.

Given a clique of size $l$ in such a joint exclusivity graph, the condition for it to violate the E-principle then becomes:
\begin{equation}
    \frac{l}{2^k} > 1 .
\end{equation}
Therefore, the existence of a clique of size $2^k + 1$ inside the joint exclusivity graph of $k$ independent $n$-cycle PR boxes guarantees violation of the E-principle.

\section{Result highlights}
\label{sec:result_highlights}
So far, we have established a connection between Ramsey theory and activation effects of the E-principle by refining the conventional method for studying the latter.
This refinement consists in introducing edge colors in the joint exclusivity graph of independent correlations to track which components contribute to the exclusivity between any pair of joint events.
Witnessing violations of the E-principle then amounts to finding cliques within the resulting edge-colored multigraph where the probabilities of the corresponding events sum to more than $1$.
Ramsey theory forces the existence of certain monochromatic structures within sufficiently large edge-colored cliques, and can thus be leveraged to study activation effects.

We now exploit the power of this general connection in the specific setting of $n$-cycle PR boxes.
In this section, we highlight the key results and the methods used to obtain them, leaving detailed proofs for \Cref{sec:Proofs}.

We fix the number of copies $k$ for which we explore activation of the E-principle by $n$-cycle PR boxes.
Recall that for $k$ independent copies of an $n$-cycle PR box, the condition for violating the E-principle is the existence of a $K_{2^k+1}$ within the edge-colored joint exclusivity multigraph $\Gamma^k$.
In \Cref{ssec:impossible} we show that one can easily construct cliques of size $2^k$ within $\Gamma^k$, but that these specific cliques can never be extended to one of size $2^k + 1$, underlining the fact that finding a $K_{2^k + 1}$ is not a straightforward task.

\Cref{tab:results_table} summmarizes the results we obtained relating $k$ and $n$, which we detail in the following subsections.
Note that a violation of the E-principle with $k$ copies guarantees a violation with $k+1$ copies, as highlighted in
\Cref{Remark1}. 

\subsection{$k = 1$}
\label{ssec:k_1}
For every $n$-cycle scenario, as for any bipartite Bell scenario, it is known that all the extremal vertices of the no-disturbance polytope satisfy the E-principle (with a single copy).
Since each possible event in an $n$-cycle PR box occurs with probability $\sfrac{1}{2}$, violating the E-principle would require the existence of a $K_3$ within the exclusivity graph of these extremal correlations.
\Cref{Corollary1} highlights the impossibility of such a $K_3$.
Therefore, activation effects become relevant for these scenarios.

\subsection{$k = 2$}
\label{ssec:k_2}
We show in \Cref{theorem3} that, except for the cases $n =4$ and $n=5$, no other $n$-cycle PR box violates the E-principle with two copies.

Recall that a $K_5$ within the joint exclusivity graph $\Gamma^2$ is required to exhibit a violation of the E-principle with two copies.
The proof of this theorem 
exploits the fact that the edges of a $K_5$ graph can be divided into two disjoint copies of $C_5$, \ie two cycles of length $5$; see \Cref{fig:K5_bifurcation}.
For $n=4$ or $n=5$, each of these $5$-cycles can be realized independently within the exclusivity graph of a single $n$-cycle PR box.
However, it is not possible to find a $C_5$ within the exclusivity graph of an $n$-cycle PR box when $n\geq 6$, as shown in \Cref{Corollary2}.
Our proof goes further to show that there is no other way to come up with a $K_5$ within $\Gamma^2$ for $n \geq 6$.

$K_5$ turns out to be a small enough clique that it can be explored analytically on a case-by-case basis to reach our results about $k=2$ for any $n$.
Still, a known Ramsey-theoretic result highlights an interesting fact about the $n =4,5$ cases, while another less well-known Ramsey result corroborates the result we obtained for $n \geq 6$, providing a neat alternative perspective.

The first of these is the classical Ramsey theory example $R(C_3,C_3) = 6$.
It implies that if a $K_6$ exists within $\Gamma^2$, one of the monochromatic projections -- and so, by \Cref{corollary:smallestoddcycle}, also the exclusivity graph of a single $n$-cycle PR box -- must contain a $C_3$ cycle.
In turn, \Cref{Corollary1} tells us that this is impossible for any $n$-cycle scenario.
This means that two copies of $n = 4,5$-cycle PR boxes can only violate the E-principle minimally, \ie with a $K_5$.
This aligns with the computational findings in Ref.~\cite{fritz2013local}, where
the search for E-principle violations in the case $k=2$ and $n=4$ (two copies of the CHSH PR box) turned up only cliques of size $5$.
This Ramsey-theoretic result explains why. 

To corroborate our result for two copies of $n\geq 6$-cycle PR boxes one can invoke a result by Erd\"os et al.~\cite{erdos1976generalized}, from one of the first papers that studied Ramsey numbers for more than two colors.
In this work, the authors define a new type of Ramsey number, $R (\leq [C])$, better explained in the proof of \Cref{theorem4} ahead.
A particular instance is $R (\leq (C_5,C_3)) = 5$;
it means that any two edge-colored $K_5$ must contain a monochromatic odd cycle of length $3$ or $5$.
\Cref{Corollary1,Corollary2} inform us that this is impossible for $n\geq 6$-cycle PR boxes, corroborating the result we obtain in \Cref{theorem3}.
This makes $n \geq 6$-cycle boxes the first reported instances in the E-principle literature of correlations that do not violate the E-principle with only two copies.

\begin{table}
    \begin{tabular}{|c|c c c c c c c c c c c c c|}
    \hline
    \multicolumn{1}{|c|}{k} & \multicolumn{13}{|c|}{n} \\
    \hline
        1 & \textcolor{red}{4} & \textcolor{red}{5} & \textcolor{red}{6} & \textcolor{red}{7} & \textcolor{red}{8} & \textcolor{red}{\dots} &  &  &
         &  &
        &  & \\
        \hline
        2 & \textcolor{blue}{4} & \textcolor{blue}{5} & \textcolor{red}{6} & \textcolor{red}{7} & \textcolor{red}{8} & \textcolor{red}{\dots}  &  &  &  &  &
         &  &  \\
        \hline
        3 & \textcolor{blue}{4} & \textcolor{blue}{5} & \textcolor{red}{6} & \textcolor{red}{7} & \textcolor{red}{8} & \textcolor{red}{\dots} &  &  &
         &  &
         &  &  \\
        \hline
        4 & \textcolor{blue}{4} & \textcolor{blue}{5} & \textcolor{gray}{6} & \textcolor{gray}{7} & \textcolor{gray}{8} & \textcolor{gray}{\dots} & \textcolor{gray}{17} & \textcolor{red}{18} & \textcolor{red}{19} & \textcolor{red}{\dots} &
         &  & \\
        \hline
        l ($\geq 5$) & \textcolor{blue}{4} & \textcolor{blue}{5} & \textcolor{gray}{6} & \textcolor{gray}{7} & \textcolor{gray}{8} & \textcolor{gray}{\dots} & \textcolor{gray}{17} & \textcolor{gray}{18} & \textcolor{gray}{19} & \textcolor{gray}{\dots}  & \textcolor{gray}{$2^l$ + 1} & \textcolor{red}{$2^l$ + 2} & \textcolor{red}{\dots} \\
        \hline
    \end{tabular}
    \caption{Relationship between $k$ and $n$ for the activation effects of E-principle with $n$-cycle PR boxes. For a given $k$ number of independent copies on the left: on the right, red marks that the $n$-cycle PR box does not violate the E-principle, blue  indicates that there is a violation, while gray indicates no knowledge of violation hence an open problem for future work.
    Results known prior to this work were only those in the first column ($n=4$).}
    \label{tab:results_table}
\end{table}

\subsection{$k = 3$}
\label{ssec:k_3}

In \Cref{theorem4}, we show that, for $n \geq 6$, $n$-cycle PR boxes do not violate the E-principle with three copies.
On the other hand, for $n=4,5$, the known two-copy violation implies a three-copy one, as per \Cref{Remark1}. 
This completely characterizes the case $k=3$ of three copies of $n$-cycle PR boxes.

Recall that a $K_9$ within $\Gamma^3$ is required to exhibit a violation of the E-principle with three copies.
In contrast with the previous $k=2$ setting, a case-by-case analysis becomes impractical.
Our proof relies on a Ramsey-theoretic result to rule out the existence of a $K_9$ within $\Gamma^3$.
A result from Ref.~\cite{erdos1976generalized} asserts that $R(\leq (C_5,C_5,C_5)) = 9$, which means that any three edge-coloring of $K_9$ contains a monochromatic odd cycle of length $3$ or $5$.
A $K_9$ within $\Gamma_3$ would then imply the existence of an odd cycle of length $3$ or $5$ in at least one of its monochromatic projections, and consequently, by \cref{corollary:smallestoddcycle} on the exclusivity graph of an individual $n$-cycle PR box. This, in turn, is impossible for $n \geq 6$ due to \Cref{Corollary1,Corollary2}.

\subsection{$k \geq 4$}
\label{ssec:k_4}
Finally, we obtain a general result applicable to an arbitrary number of copies $k$, bounding the size $n$ of $n$-cycle PR boxes that may witness a $k$-copy activation effect.
Specifically, in \Cref{theorem5} we show that for $n \geq 2^k + 2$, $n$-cycle PR boxes do not violate the E-principle with $k$ copies.
Since we have already fully characterized the cases $k = 2,3$, this result becomes relevant for $k \geq 4$.
This leaves open the range $6 \leq n \leq 2^k + 1$, for which we do not know whether $k$ copies suffice to violate the E-principle.

Recall that a $K_{2^k + 1}$ within $\Gamma^k$ is required to exhibit violation of E-principle with $k$ copies.
Our proof applies a result mentioned in the introduction of Ref.~\cite{day2017multicolour}, a recent work that explored multi-colored Ramsey numbers for odd cycles. It states that a $k$ edge-colored $K_{2^k+1}$ must contain a monochromatic odd cycle.
We then leverage \Cref{Theorem 2}, which provides the size of the smallest odd cycle in the exclusivity graph of any $n$-cycle PR box, to obtain the bound.

\section{Proofs}
\label{sec:Proofs}
We now provide the proofs of the results highlighted in the previous section.

\subsection{Exclusivity graphs of n-cycle PR boxes}
\label{ssec:E-graphs_and_properties}
We start with the characterization of the exclusivity graphs of $n$-cycle PR boxes.
As mentioned in \Cref{ssec:n-cycle_PRboxes},
these turn out to be well-studied graphs:
the exclusivity graph of an $n$-cycle PR box takes the form of a M\"obius ladder graph ($M_{2n}$) when $n$ is even and of a Prism graph ($Y_n$) when $n$ is odd.
\Cref{fig:E-graph} illustrates these two situations.
Note that the standard notation for Möbius ladder graph, $M_{2n}$, indicates its order (total number of vertices) in the subscript, while the standard notation for the circular ladder graph, $Y_n$, indicates only half of its order.

\begin{figure}
    \centering
    \includegraphics[scale = 0.5]{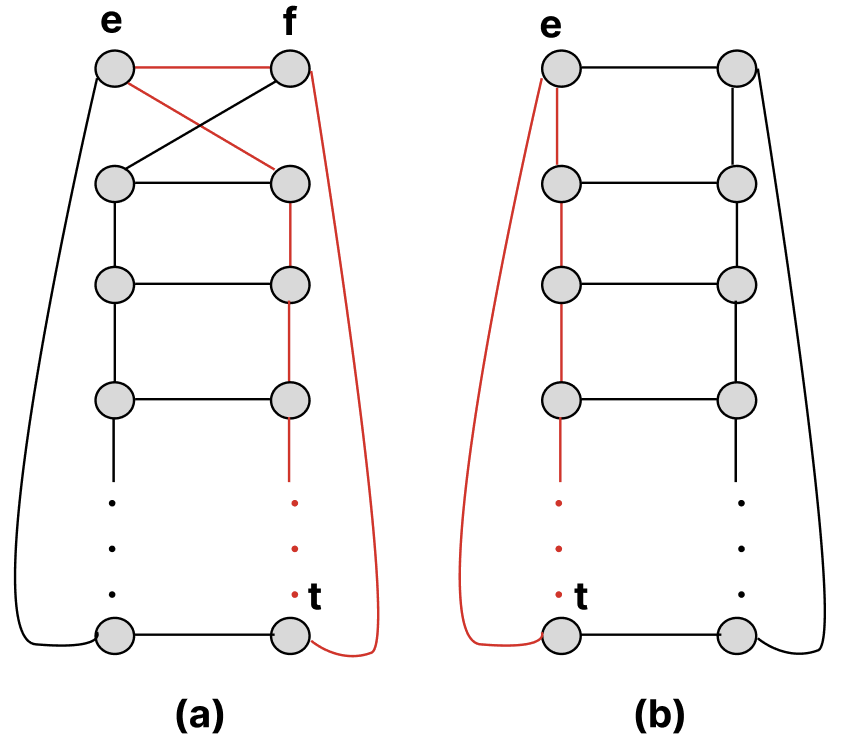}
    \caption{Exclusivity graphs of the $n$-cycle PR boxes: (a) when $n$ is even: Möbius ladder graph ($M_{2n}$) (b) when $n$ is odd: circular ladder graph, also popularly known as Prism graph ($Y_n$). Drawn in red are the smallest odd cycles for the two cases where $e,f,t$ are the events used in \Cref{Theorem 2} to illustrate the construction of these cycles.}
    \label{fig:E-graph}
\end{figure}
\begin{figure}
    \centering
    \includegraphics[scale = 0.9]{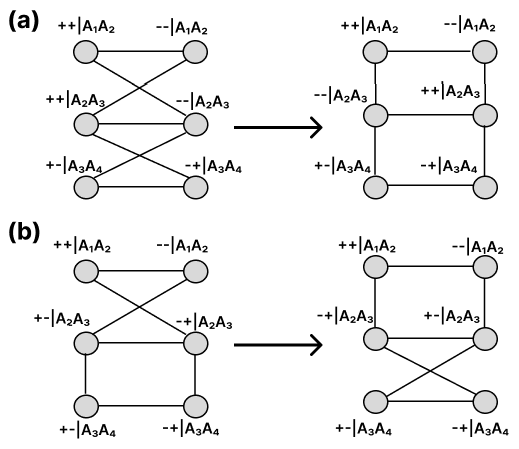}
    \caption{Interchanging nodes of the middle context (a) both consecutive criss-cross connections become vertical (b) a cross connection and a vertical one beneath it turn into vertical and cross connections respectively.}
    \label{fig:interchange}
\end{figure}
\begin{theorem}
\label{theorem1}
    The exclusivity graph of an $n$-cycle PR box is isomorphic
    to the Möbius ladder graph $M_{2n}$ when $n$ is even
    and to the circular ladder graph $Y_n$ when $n$ is odd.
\end{theorem}
\begin{figure}
    \centering
    \includegraphics[width=0.5\linewidth]{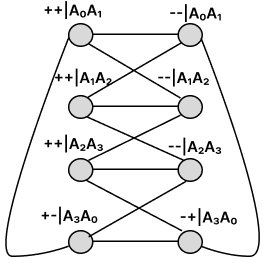}
    \caption{Reference configuration for the exclusivity graph of PR-box. Positively correlated events of a given context, when the joint outcome $++$ is on the left, connect in a cross-like manner to the context under it while negatively correlated events of a given context, when the joint outcome $+-$ is on the left, arrange in a vertical manner to the context under it.}
    \label{fig:configuration}
\end{figure}
\begin{proof}
Consider a specific starting configuration of the exclusivity graph of an $n$-cycle PR box. The configuration starts from the possible events of first context at the top and go down to the final context at the bottom in the cyclic order of contexts. For a given context, if the possible events are perfectly correlated then the node on the left will denote the outcome $++$ while the one on the right denotes outcome $--$ ; while when the possible events are anti-correlated, the node on the left will denote outcome $+-$ while the one on right will denote outcome $-+$. \Cref{fig:configuration} illustrates this reference configuration for a 4-cycle PR-box. Once this arrangement is fixed we know how nodes of a given context connect with the nodes of the context underneath it. Note that since the ordering of the contexts is cyclic, a given context has a measurement in common only to a context before and after it. Hence, no exclusivity relation exists between any nodes beyond the adjacent contexts in the cycle. Within this configuration, the nodes corresponding to perfectly correlated events of a given context connect to the nodes of the context (cyclically) underneath it in a cross like manner irrespective of whether the context has correlated or uncorrelated events. On the other hand, when the nodes, of a given context, correspond to anti-correlated events they connect with the context (cyclically) beneath it in a vertical manner irrespective of whether the context has correlated or uncorrelated events.

With this initial configuration of the exclusivity graph in hand, we now try to convert the cross connections into vertical ones by interchanging the nodes within contexts, in a sequential manner, to redraw it -- hence exhibit isomorphism -- to the Möbius ladder graph and the circular ladder graph for $n$ even and odd cases respectively. Note the obvious but important fact that each cross connection is followed either by a cross connection or a vertical one in this configuration, as shown on the left sides of \Cref{fig:interchange}(a) and \Cref{fig:interchange}(b).

Let us start from the very first pair of contexts leading to a cross connection in this initial configuration. To convert this cross connection into a vertical one, we interchange the nodes in the (cyclically) bottom context (i.e. events of the bottom context in the pair) while keeping the nodes of the first context as they were. This means that, if before this interchange, there was a cross connection right underneath the first cross connection (l.h.s of \Cref{fig:interchange}(a)) then after the interchange there are only vertical connections across the three contexts as is shown in the r.h.s of \Cref{fig:interchange}(a). On the other hand if before the interchange there was a vertical connection beneath the first cross connection (l.h.s of \Cref{fig:interchange}(b)) then after the interchange this vertical connection converts into a cross connection as shown in the r.h.s of \Cref{fig:interchange}(b).

Overall, in the first case we get rid of the two cross connections, converting them into vertical ones, while in the latter case it is as if the first cross connection has been transported (cyclically) downwards into the graph. 
Note that when $n$ is even there will be odd number of cross connections in the initial configuration (since by definition there are odd positively correlated contexts) while when $n$ is odd there are even number of cross connections in the initial configuration. Repeating the interchange process described above sequentially leads to having only one cross connection when $n$ is even while all cross connections will be transformed into vertical ones, when $n$ is odd. Overall, this leads to a Möbius ladder graph ($M_{2n}$) when $n$ is even, see \Cref{fig:E-graph}(a), while a circular ladder graph ($Y_n$) (also known as Prism graph), when $n$ is odd, see \Cref{fig:E-graph}(b).
\end{proof}
A single copy of an $n$-cycle PR box satisfies E-principle. This is because for an $n$-cycle PR box violation of E-principle requires existence of $K_3$ inside the exclusivity graph. This is not possible as captured by the following Corollary of the above theorem.
\begin{corollary}
\label{Corollary1}
The exclusivity graph of an $n$-cycle PR-box does not contain a triangle ($K_3$) for any $n \geq 4$.
\end{corollary}
\begin{proof}
As can be seen in \Cref{fig:E-graph}, in each graph $M_{2n}$ and $Y_n$ (for any $n$), the nodes in the neighbourhood of a given node are not connected with each other. Therefore, there is no $K_3$ possible inside these graphs.  
\end{proof}
Now that we know how the exclusivity graphs of $n$-cycle scenarios look like we are ready to define the most relevant property of these graphs for our proofs later. The property is the size of the smallest odd cycle inside these graphs. The following theorem captures that.
\begin{theorem}
\label{Theorem 2}
The smallest odd cycle in the exclusivity graph of an $n$-cycle PR box has length $n + 1$ when $n$ is even and length $n$ when $n$ is odd.
\end{theorem}
\begin{proof}
We first show an explicit construction of the respective odd cycles mentioned in the theorem statement and then provide a unified proof for the two cases exhibiting why there can't be a smaller odd cycle in the two cases.

Explicit construction of respective odd cycles: pick a path from the first context to the last where only one node is picked from each context. Say you pick event $e$ from the first context and end up at some event $t$ in the last context from top to bottom. Now, to close this path into a cycle we need to go back from event $t$ to the first context. When $n$ is even we can return back to event $f \neq e$ in the first context (because event $t$ is connected to $f$), hence giving an odd cycle of size $n + 1$, see  \Cref{fig:E-graph}(a). In the case when $n$ is odd, event $t$ straight away connects with event $e$ and hence we get an odd cycle of size $n$, see  \Cref{fig:E-graph}(b). 

Now we show that a smaller odd cycle can't exist in the two cases. Note that for a smaller odd cycle to exist we can't use events from all the contexts. More explicitly, when $n$ is even i.e. the graph is $M_{2n}$, a smaller odd cycle would have a size less than $n+1$, which means the biggest such cycle can only be of size $n-1$. When $n$ is odd, the biggest such odd cycle would be of size $n - 2$. Since there are $n$ contexts, we cannot pick at least one event from each context to construct such odd cycles. This means that such an odd cycle has to live inside a sub graph utilizing less than $n$ contexts in both the cases. This further implies that such a sub graph would either be a ladder graph or a sub graph of a ladder graph for both cases. It is well known that a ladder graph (hence any of its sub graph as well) is a bi-partite graph. Furthermore, it was first shown by König that a graph is bi-partite if and only if it contains no odd cycle \cite{konig1990theory,koh2007introduction}. This proves that an odd cycle less than $n+1$ for $M_{2n}$ when $n$ is even and less than $n$ for $Y_n$ when $n$ is odd can't exist. This finishes our proof.
\end{proof}
\begin{corollary} 
\label{Corollary2}
The exclusivity graph of an n-cycle PR-box contains a $5$-cycle $C_5$ only if $n=4$ or $n=5$.
\end{corollary}
\begin{proof}
From \Cref{Theorem 2}, when $n = 4$ or $n=5$, the smallest odd cycle has size 5. For $n \geq 6$, it is always larger than 5.
\end{proof}
\begin{figure}
    \centering
    \includegraphics[scale = 0.4]{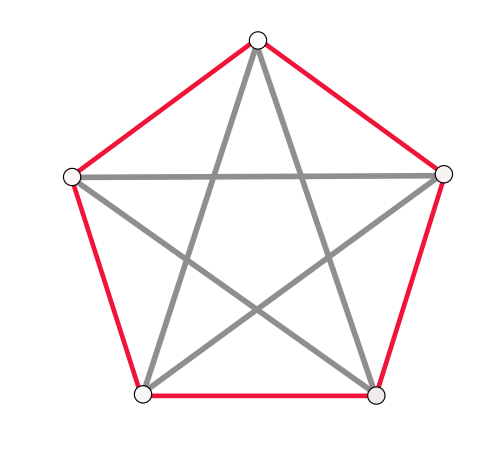}
    \caption{Edge bifurcation of a $K_5$ exclusivity graph into two edge-wise disjoint $C_5$ graphs. The outer one denoted in red and the inner one in grey.}
    \label{fig:K5_bifurcation}
\end{figure}
\subsection{Impossibility of extending a trivial $K_{2^k}$ to $K_{2^k +1}$}
\label{ssec:impossible}
Before we move on to studying activation effects of $n$-cycle PR boxes we eliminate a trivial possibility of violation of E-principle. We know that for an $n$-cycle PR box, the condition of violation of E-principle is the existence of $K_{2^k+1}$ inside the $k$ edge-colored joint exclusivity multigraph ($\Gamma^k$). Here we show that there always exists a $K_{2^k}$ inside the $k$ edge-colored $\Gamma^k$ that can't be extended to a $K_{2^k + 1}$ supporting the fact that finding a $K_{2^k + 1}$ isn't a straightforward task. 

It is easy to visualize that OR product of two complete graphs $K_n$ and $K_m$ is a complete graph $K_{m+n}$. This implies that inside $\Gamma^k$ exists sub graph $K_{2^k}$ which can be seen as an OR product of $k$, $K_2$ graphs. Such a $K_{2^k}$, let's call it $T$, can never be extended to a $K_{2^k+1}$ inside $\Gamma^k$. We now prove this claim: Let us see if some node $v$ exists in $V(\Gamma^k) \setminus V(T)$ such that $v$ and $V(T)$ together produce a clique of size $2^k + 1$. Let us say that $v$ is exclusive to an event $e \in V(T)$ due to its $lth$ component i.e. $(e_l,v_l) \in E(G)$. For each event in $V(T)$ the $lth$ component is either $e_l$ or $e_{l}^{\perp}$. In fact, half of the events in $V(T)$ have $e_l$ as their $lth$ component while others have $e_{l}^{\perp}$. If $(e_l,v_l) \in E(G)$ and we know, by construction that $(e_l,e_{l}^{\perp}) \in E(G)$, then $(e_{l}^{\perp},v_l) \notin E(G)$ because there can't be any $K_3$ in $G$ as shown in Corollary 1. Therefore, a connection due to the $lth$ component would divide the set of events in $V(T)$ into two equal subsets, say $A$ and $B$ each of size $2^{k-1}$, where all events in $A$ are connected with $v$ due to the $lth$ component being $e_l$ while $v$ can only connect with the nodes in $B$ because of components other than the $lth$ one. Now, say $v$ connects with a one node in $B$ due to the $mth$ component where $m \neq l$. Due to this the nodes within $B$ now get divided into two equal disjoint subsets with each subset having size $2^{k-2}$. Progressing this way, for each new component of $v$, we are left with a subset of nodes which are unconnected to $v$ due the current component under consideration with half the size of the previous set. Therefore, when we have utilized all the components of $v$, we will be left with a subset of size $2^{k-k} = 1$ that will remain disconnected from $v$ suggesting that we can never form a $K_{2^k + 1}$ using $v$ and $V(T)$. Therefore, such a $K_{2^k}$ can't be extended to $K_{2^k + 1}$. We also state a trivial consequence of the fact, that an OR product of $K_m, K_n$ will be a $K_{m+n}$, mentioned earlier above as a Remark.
\begin{Remark}
\label{Remark1}
Given a clique of size $s$ in the joint exclusivity graph of $k$ copies $\Gamma^k$, one can find a clique of size $2s$ inside $\Gamma^{k+1}$.   
\end{Remark} This holds because one can do an OR product of $K_s$ inside $\Gamma^k$ with a $K_2$ inside the $(k+1)$th copy of the original graph. This produces a $K_{2s}$ inside $\Gamma^{k+1}$.

\subsection{Activation effects of $n$-cycle PR boxes}
\label{ssec:activation_eff}
With the trivial case out of our way we are now ready to provide our results concerning violation of E-principle for different $k$. Simplest case being $k = 2$ where a $K_5$ is desired inside $\Gamma^2$ to witness violation of E-principle. 
\begin{theorem}
\label{theorem3}
   The joint exclusivity graph of two independent $n$-cycle PR-boxes contains a $K_5$ when $n = 4$ or $n=5$. That is, two copies activate violation of the E-principle for $4$- and $5$-cycle PR boxes. 
\end{theorem}
\begin{figure}
    \centering
    \includegraphics[scale = 0.4]{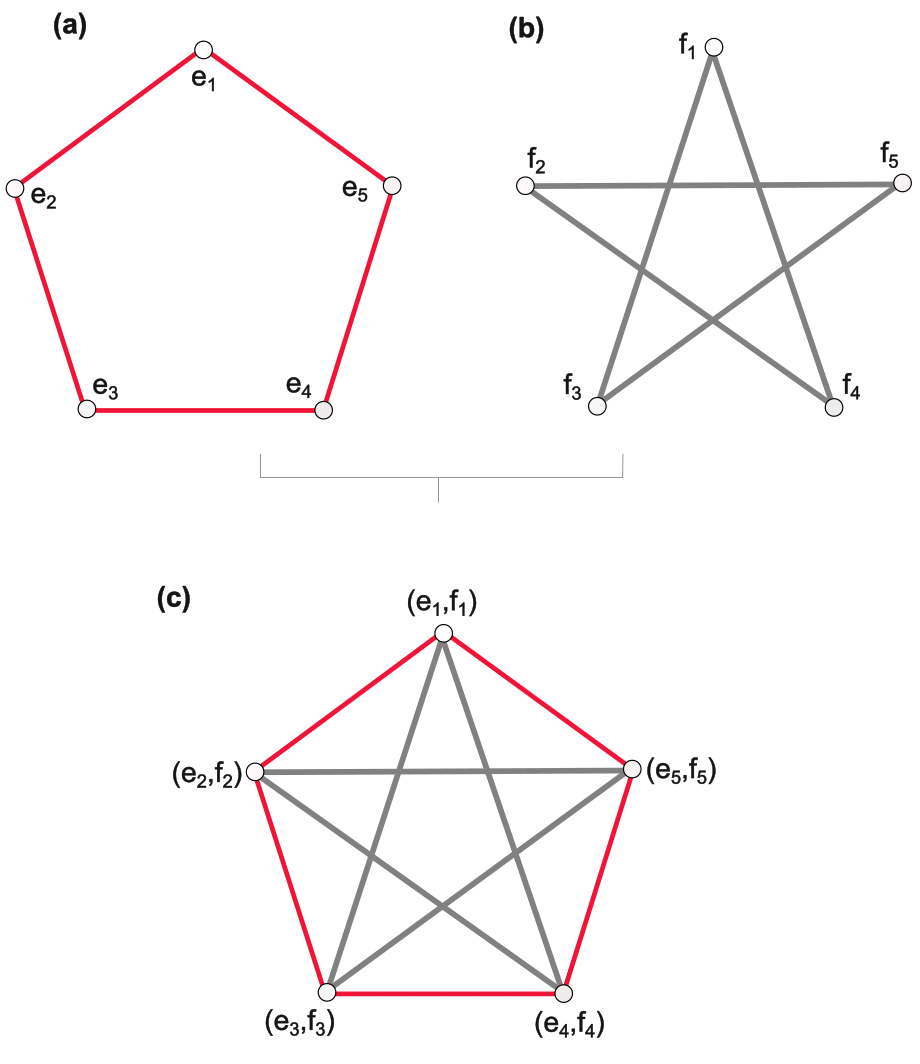}
    \caption{Constructing a $K_5$ using two independent $C_5$ cycles via multicolor product: (a) events coming from PR-box whose exclusivity graph is colored in red (b) events coming from PR-box whose exclusivity graph is colored in grey (c) joint events from the two copies leading to a $K_5$.}
    \label{fig:K_5_with_n_45}
\end{figure}
\begin{proof}
We start with a nice observation that all the edges of a $K_5$ can be divided into two disjoint subsets. The outer edge-induced graph, denoted with red in \Cref{fig:K5_bifurcation}, denotes a $C_5$ while the inner one, denoted with grey in \Cref{fig:K5_bifurcation}, is also isomorphic to $C_5$. The independence of events at each node of $K_5$ along with this edge based bifurcation of $K_5$, means that one can use \Cref{Corollary2} to ensure that the exclusivity in the outer and the inner $C_5$ graphs is realized independently by two different copies of the $n$-cycle PR-boxes when $n=4,5$. $C_5$ in \Cref{fig:K_5_with_n_45}(a) denotes events $\{e_i\}_{i=1}^5$ coming from one copy while events $\{f_i\}_{i=1}^5$ in \Cref{fig:K_5_with_n_45}(b) are coming from the other PR-box. One can then construct joint independent events $\{(e_i,f_i)\}$ whose exclusivity graph is a $K_5$ as denoted in \Cref{fig:K_5_with_n_45}(c). This means E-principle will be violated for $n = 4,5$ by just using two copies.
\begin{figure}
    \centering
    \includegraphics [scale = 0.4]{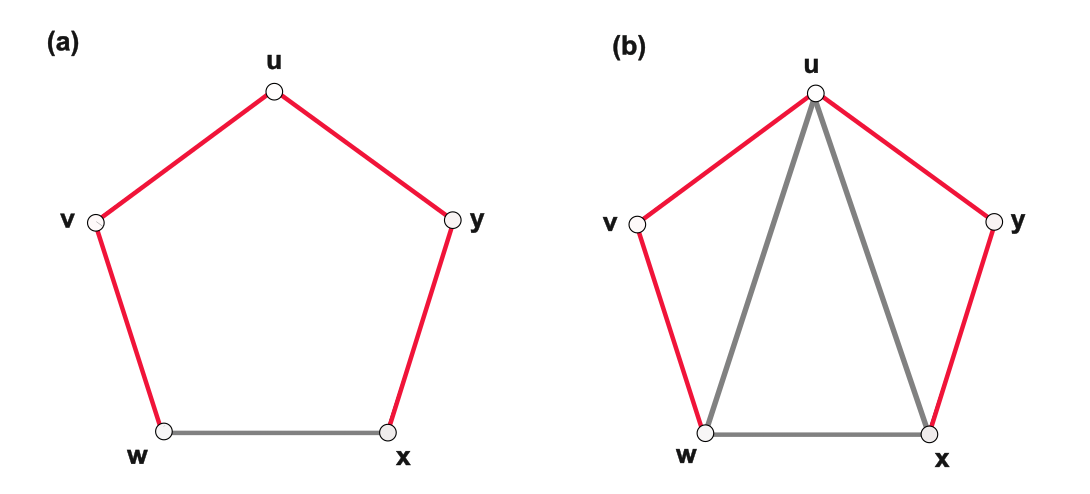}
    \caption{Impossibility of extending a $C_5$ in (a) with one edge supported by $n$-cycle PR-box whose exclusivity graph is colored grey while the rest of the edges are supported by the $n$-cycle PR box whose exclusivity graph is colored red, to (b) a $K_5$ for all $n \geq 6$ cycle scenarios since it is inconsistent with \Cref{Corollary1}.}
    \label{fig:One_S1_C_5}
\end{figure}
Notice that \Cref{Corollary2} implies that such a construction of $K_5$ won't work for $n \geq 6$. We now show that there is also no other way to obtain a $K_5$ with two copies when $n \geq 6$.
\end{proof}
\begin{theorem}
\label{theorem6}
    The joint exclusivity graph of two independent $n$-cycle PR-boxes does not contain a $K_5$ for $n\geq 6$. That is, two copies do not activate a violation of the E-principle for $n \geq 6$-cycle PR boxes.
\end{theorem}
\begin{proof}
\begin{figure}
    \centering
    \includegraphics[scale = 0.4]{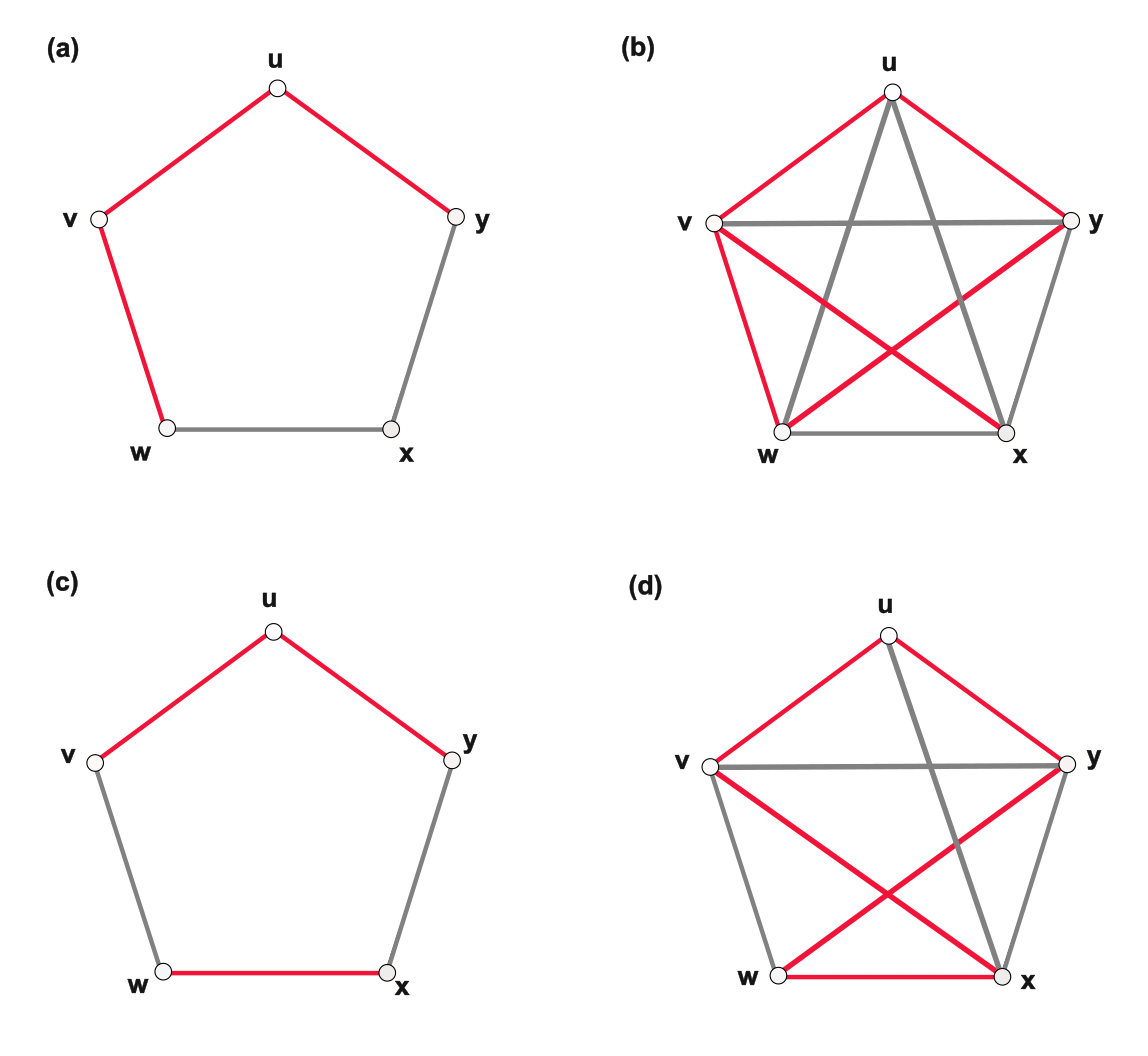}
    \caption{Impossibility of extending a $C_5$, with two edges supported by events from an $n$-cycle PR box whose exclusivity graph is colored grey while the other edges are supported by an  $n$-cycle PR box whose exclusivity graph is colored red, to a $K_5$: (a) when the two grey edges are adjacent, the extension leads to violation of \Cref{Corollary1} as shown in (b) with nodes $\{u,w,x\}$; (c) when the two grey edges are not adjacent, the extension leads to violation of \Cref{Corollary2} due to a $C_5$ with nodes $\{u,v,x,w,y\}$ as shown in (d).}
    \label{fig:two_S2_K_5}
\end{figure}
If a $K_5$ were to exist for two $n \geq 6$ PR-boxes, thanks to \Cref{Corollary2}, a $C_5$ inside the joint exclusivity structure $\Gamma^2$ cannot have all its edges supported by events from just one PR box. In other words, all edges cannot be assigned the same color. This leads to two cases:
\begin{enumerate}
    \item where a $C_5$ has exactly one edge that is colored grey while the rest are colored red, see \Cref{fig:One_S1_C_5}(a).
    \item where a $C_5$ has exactly two edges that are colored grey while the rest are colored red. This leads to two situations based on whether the two grey edges are adjacent or not in a $C_5$, see \Cref{fig:two_S2_K_5}(a),(c).   
\end{enumerate}
Consider case 1, see \Cref{fig:One_S1_C_5}(a), with just one edge with color grey. \Cref{Corollary1} implies that in the $C_3$ composed of events $\{u,v,w\}$ the edge connecting events $u$ and $w$ must be colored grey. The same holds true for the $C_3$ composed of $\{u,x,y\}$ where the edge corresponding to events $\{u,x\}$ is also forced to be of color grey, see \Cref{fig:One_S1_C_5}(b). This leads to a $C_3$ graph $\{u,w,x\}$ where all edges are of grey color. This $C_3$ violates \Cref{Corollary1}, hence case 1 is impossible.
With a similar sequential extension of graphs via \Cref{Corollary1} as above we can show that case 2 is also impossible. \Cref{fig:two_S2_K_5} covers both situations mentioned in case 2 above. Starting from the sub case represented by \Cref{fig:two_S2_K_5}(a), using \Cref{Corollary1} to sequentially extend each of the graphs with vertex sets $\{u,v,y\}$, $\{w,x,y\}$, $\{u,v,w\}$, $\{x,v,w\}$, and
$\{u,v,x\}$ into a $C_3$ leads to a $K_5$ depicted in \Cref{fig:two_S2_K_5}(b) where thee $C_3$ composed of events $\{u,w,x\}$ violates \Cref{Corollary1}.

Starting from the sub case represented in \Cref{fig:two_S2_K_5}(c), using \Cref{Corollary1} to sequentially extend each of the graphs with vertex sets $\{u,v,y\}$, $\{v,x,y\}$, $\{u,v,x\}$, and $\{w,v,y\}$ into a $C_3$ leads to a $C_5$ composed of $\{u,v,x,w,y\}$ with all its edges having color red, see \Cref{fig:two_S2_K_5}(d), violating \Cref{Corollary2}. This exhausts the possibility of violation of E-principle with two copies of $n$-cycle PR boxes $ \forall \; n \geq 6$.
\end{proof}

\begin{corollary}
    For $n \geq 6$, the joint exclusivity graph of a pair of $n$-cycle PR-boxes does not contain a $K_l$ for any $l \geq 6$. 
\end{corollary}
\begin{proof}
For a pair of n-cycle PR-boxes of a given $n \geq 6$, a $K_l$ exclusivity graph for $l\geq 6$ would contain a $K_5$ as its induced sub graph but the latter can't exist due to \Cref{theorem3}, hence, proving the corollary.
\end{proof}
%
\subsection{Exploiting Ramsey theory}
\label{ssec:exploiting_ramsey_theory}
Until now we have seen that $n = 4,5$ cycle PR boxes show activation effects as soon as more than one copy is considered while $n \geq 6$ cycle PR boxes don't when two copies are taken into account. We now see how our new found connection with Ramsey theory can help us say more about activation effects of $n$-cycle PR boxes. 

The first result follows from the canonical Ramsey number $R(K_3,K_3) = 6$. For $n$-cycle PR boxes this means that for two copies if a $K_6$ exists inside $\Gamma^2$, it will have to have at least one monochromatic $K_3$ as sub graph. This isn't possible, thanks to \Cref{Corollary1}, for any $n$. For $n = 4,5$ this means that we can't violate E-principle beyond $K_5$ i.e. $R(K_3,K_3) = 6$ upper bounds violation of E-principle in $\Gamma^2$. This explains why in Ref.~\cite{fritz2013local} authors didn't find a clique bigger than $K_5$ for two copies of CHSH ($n = 4$ cycle) PR box. 

We now utilize known Ramsey theory results to extract useful results for $n \geq 6$ cycle PR boxes. Our first result follows from one of the many small but significant results proven in \cite{erdos1976generalized}. This work from Paul Erdös et. al. is one of the first works in the Ramsey theory literature that considered going beyond two colorings of the edges. The authors considered three and four colorings of edges where the structures assigned to the colors consisted of complete graphs, complete bipartite graphs, paths, and cycles. The relevant result for us is when, for three possible colorings of edges, each color is assigned an odd cycle. The authors define `cycle Ramsey number' $R(\leq [C]) = p$ (where $[C] \equiv (C_{2d_1+1},C_{2d_2+1},C_{2d_3+1})$ which is to be read as color $i$ being assigned a cycle of size $2d_i+1$) which denotes the smallest integer $p$ such that for any arbitrary three coloring of $K_p$ there exists at least one $i$ that contains an odd cycle of size less than or equal to $2d_i+1$.
\begin{theorem}
\label{theorem4}
    For $n \geq 6$, the joint exclusivity graph of three independent $n$-cycle PR boxes does not contain a $K_9$. That is, three copies cannot activate a violation of the E-principle for $n \geq 6$-cycle PR boxes.
\end{theorem}
\begin{proof}
$R(\leq (C_5,C_5,C_5)) = 9$ obtained in \cite{erdos1976generalized} translated in the language of our work means that for a $K_9$ to exist in a joint exclusivity graph constructed from three statistically independent copies of $n$-cycle PR boxes (for any $n \geq 4$), an individual copy itself must have a $C_5$ or a $C_3$ cycle as its sub graph. But from \Cref{Theorem 2}, we know that for any $n \geq 6$ cycle PR box, the smallest odd cycle has size more than 5 implying that a $K_9$ can't exist in such cases. This finishes the proof of \Cref{theorem4}. This marks the impossibility of violating E-principle using three copies of $n \geq 6$-cycle PR boxes.
\end{proof}
This theorem provides the first example of extremal non-classical correlations for which three copies aren't enough to witness violation of E-principle. Earlier, for the scenarios explored, in \cite{fritz2013local}, two copies were always enough to witness E-principle's violation.

Another result from Ref.~\cite{day2017multicolour} exploring edge coloring of $K_{2^k + 1}$ graphs with $k$ colors in the context of Ramsey theory helps us infer the class of $n$-cycle scenarios whose PR boxes can't violate E-principle for a given $k$, where $k \geq 4$. The result states that any $k$ edge-colored clique of size $2^k + 1$ must have a monochromatic odd cycle. For the ease of reader we explain why that must be the case. It is easy to show that for complete graph $K_{2^k}$ there exists a $k$ coloring of edges such that each color comprises a bipartite graph. Before we provide an explicit construction for this we show that such $k$ edge colorings exist for $K_m$ only if $m \leq 2^k$. 
\begin{figure}
    \centering
    \includegraphics[scale = 0.8]{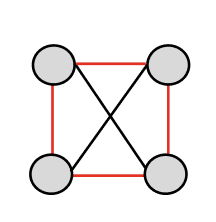}
    \caption{Two edge-colored $K_4$ where each monochromatic sub graph is a bi-partite graph.}
    \label{fig:bi-partite}
\end{figure}

Consider a $k$ edge coloring of $K_m$ where each coloring corresponds to a bi-partite graph. For each vertex of this graph $K_{m}$, imagine writing a binary vector of length $k$ where the $i{th}$ coordinate determines which side of the bi-partition of color $i$ this vertex lies. All vertices of $K_m$ must receive distinct binary vector because if two nodes receive the same binary vector it means there is no edge that connects the two nodes. Since there can be at most $2^k$ such distinct labels, this means that $m \leq 2^k$. Therefore for $m > 2^k$ no such coloring exists. One explicit example of the above coloring when $m = 2^k$ is as follows: Take a two edge-colored $K_4$ such that each monochromatic projection is a bi-partite graph, see \Cref{fig:bi-partite}. Now take two copies of this graph and connect the two $K_4$ graphs via a third color. This gives a $K_8$ all of whose monochromatic projections are bi-partite. Following the same idea one can progressively build a $k$ edge-colored $K_{2^k}$, each of whose $k$ monochromatic projections is bi-partite. Following theorem captures this result's consequence for $n \geq 6$ cycle PR boxes.
\begin{theorem}
\label{theorem5}
For $n \geq 2^k + 2$, the joint exclusivity graph of $k$ independent copies of an $n$-cycle PR box does not contain a $K_{2^k + 1}$. That is, $k$ copies do not activate a violation of the E-principle for $n \geq 2^k + 2$-cycle PR boxes.
\end{theorem}
\begin{proof}
The proof follows from a simple result provided in Ref.~\cite{day2017multicolour}. The result states that any time we have a $k$ edge-colored $K_{2^k + 1}$ there then has to be at least one monochromatic graph that is not bi-partite i.e. every $k$ edge coloring of $K_{2^k + 1}$ contains an odd monochromatic cycle. A direct implication of this result is that for $K_{2^k + 1}$ to exist in $\Gamma^k$ at least one monochromatic projection must have an odd cycle. But we know from \Cref{Theorem 2} that the smallest odd cycle for $n \geq 2^k + 2$ is of size $2^k + 3$ which is bigger than ${2^k + 1}$ the order of clique that is necessary for violation. Hence, no $n \geq 2^k + 2 $ PR box can violate the E-principle for $k \geq 4$ copies. This finishes the proof.
\end{proof}
\section{Discussion}
\label{sec:discussion}

We conclude with a quick summary of the methodology used 
and a discussion of possible avenues for future research along similar lines.

\subsection{Conclusion}
\label{ssec:conclusion}

In this work we refined the conventional method for probing activation effects of the exclusivity principle. We then applied this refinement to explore activation effects for the $n$-cycle PR boxes, generating a wealth of results highlighted in \Cref{sec:result_highlights}.

Our refinement amounts to preserving in the joint exclusivity structure the information as to which components determine the exclusivity between each pair of joint events,
something not taken into account in the conventional approach.
We do so by building a joint exclusivity structure consisting of an edge-colored multigraph, where we assign a unique color to each component and use it to mark the edges (denoting exclusivity between joint events) determined by the exclusivity graph of that component (namely, by the exclusivity between the individual events at that component),
as explained in \Cref{ssec:multi_color_or_product}.

The problem of exhibiting activation of the E-principle then maps to finding specific cliques within this joint edge-colored exclusivity multigraph.
This is where Ramsey theory comes into play, as it necessitates the presence of certain monochromatic structures inside edge-colored cliques, which in turn translates to constraints on the exclusivity graphs of the individual components.
If these constraints are impossible to satisfy, then no clique of the required size exists in the joint exclusivity multigraph, ensuring that the E-principle is not violated in the given situation.
On the other hand, however, satisfying the constraints from a given Ramsey result does not guarantee the existence of the corresponding clique.
As such, Ramsey theory plays an asymmetric role in studying activation effects of the E-principle.

Luckily, in our investigation of $n$-cycle PR boxes, we found relevant results in the Ramsey theory literature
that let us draw conclusions about activation effects of the E-principle, both excluding violations of the E-principle with certain numbers of copies and bounding the amount of existing violations (as in the case $n =4,5$ with $k=2,3$).

\subsection{Outlook}
\label{ssec:outlook}

One may wonder whether these or other Ramsey-theoretic results can be used to explore E-principle activation effects for correlations on more general measurement scenarios.
An advantage of considering $n$-cycle PR boxes is that all possible events occur with the same probability.
Because of this, the problem of finding violations of the E-principle with a given number of copies maps to finding cliques of a certain size, with no consideration given to the weights of the nodes (\ie the probabilities of the corresponding joint events) due to their uniformity. 
The same property does not hold in general for (extremal boxes of) arbitrary scenarios, wherein witnessing violations also depends on the weights of the nodes.
One could still, in principle, use Ramsey-flavoured results for vertex-weighted graphs.
The main obstacle is that finding Ramsey numbers is very hard:
after 100 years of research, there are only partial results, mostly consisting of upper or lower bounds for Ramsey numbers, going only up to four colors \cite{radziszowski2012small}.
It would be interesting to investigate how these known results can be utilized in exploring violations of the E-principle in more general scenarios, as well as other questions in physics that can be reformulated as clique-finding problems.

Another natural question, complementary to those analysed in this work, is whether for every $n$-cycle PR box there exists a number of copies $k$ that activates a violation of the E-principle?
We conjecture that the answer to this question is yes.
We have already seen that two copies suffice when $n=4,5$, but not when $n \geq 6$.
In fact, we further showed that even three copies are not enough to activate a violation when $n \geq 6$. 

Let us provide some evidence in support of the conjecture.
Cabello~\cite[Lemma 4]{cabello2019quantum} showed that for a given exclusivity graph, its theta body $TH(G)$ coincides with the largest set $\mathcal{P}$ of probability assignments (vertex-weightings) such that:
(i) the independent product of any two assignments in the set (\ie $p \otimes p'$ for $p, p' \in \mathcal{P}$) satisfies the E-principle, and 
(ii) the independent product of any number of copies of an assignment in the set (\ie $p^{\otimes k}$ for $p \in \mathcal{P}$) satisfies the E-principle
\footnote{Note that the the assignment $p \otimes p'$ is a probability assignment for the OR-product of the graph $G$ with itself, while $p^{\otimes k}$ is a probability assignment for the OR-product of $k$ copies of $G$.}. 
Note that $TH(G)$ collects the probability assignments arising from all the quantum correlations consistent with that graph, not just from a fixed scenario; that is, it contains quantum correlations coming from all possible scenarios that produce events realizing the given exclusivity graph.
Using the Lovász number characterization of $M_{2n}$ and $Y_n$ from Refs.~\cite{araujo2013all,bharti2022graph},
one can easily establish that any $n$-cycle PR box lies outside the theta body of its exclusivity graph.
However, this is not quite enough to conclusively answer our question:
in light of the above-mentioned result, it is conceivable that a correlation outside the theta body satisfies the E-principle with any number of copies
but sidesteps the result due to condition (i), by violating the E-principle when considered jointly with another correlation from $TH(G)$.
This seems unlikely, especially given the highly symmetric nature of $n$-cycle scenarios.
Still, whether this is possible is left as an open problem for now.

A potentially simple start to answering the general question posed above would be to focus on the case $n = 6$ and try to construct a violation-producing clique for some number of copies $k \geq 4$.
A crude, back-of-the-envelope attempt to find a $K_{17}$ for $k = 4$ and $n = 6$ along the lines of how a $K_5$ is constructed when $k = 2$ and $n = 4, 5$ fails.
The ten edges of $K_5$, as shown in \Cref{fig:K_5_with_n_45}, are divided into two disjoint subsets such that
the edges in each subset are determined by one of the copies of the $4$- or $5$-cycle PR box, with no edge from any of the copies being used more than once (\ie no edge from a single copy determines more than one edge of $\Gamma^2$ that appears in the clique $K_5$).
A similar division of the edge set of $K_{17}$ into four disjoint subsets, in such a way that each copy handles the edges in one of these subsets without repeating the use of any of its edges, is impossible.
$K_{17}$ has $136$ edges; dividing them into four subsets would result in at least one copy being assigned $34$ edges.
However, the exclusivity graph of an $n$-cycle PR box has a total of $3n$ edges, which for $n = 6$ amounts to $18$ edges, far fewer than the required $34$.
This means that, in any potential realization of $K_{17}$ within $\Gamma^4$, edges from at least one of the copies would have to be reused multiple times to determine edges of $\Gamma^4$ (exclusivity between joint events) lying in the clique $K_{17}$.
And the problem only gets worse as the number of copies $k$ increases: for a fixed $n$, the size of a violation-producing clique $K_{2^k + 1}$ increases exponentially with $k$, with its total number of edges 
being $2^{2k-1} + 2^{k-1}$, while the number of edges of the exclusivity graph of an individual copy remains fixed as $3n$, so that the number of edges that $k$ copies could cover without repetitions grows linearly, as $3nk$.
Therefore, finding a violation-producing clique would require a much more intricate construction than the case of $K_5$, underscoring the complexity of this problem.
This further emphasizes the power and potential of our established connection to -- and exploitation of -- Ramsey theory.

We conclude hoping that the unexpected bridge established between Ramsey theory and quantum foundations encourages fruitful collaboration and two-way interaction between mathematicians and physicists in these fields.

\begin{acknowledgments}
We thank Som Kanjilal, Ad\'an Cabello, and Rafael Wagner for valuable comments and discussions. The work was supported by the Digital Horizon Europe project \href{https://doi.org/10.3030/101070558}{FoQaCiA}, \textit{Foundations of Quantum Computational Advantage}, GA no. 101070558.
The authors also acknowledge financial support from FCT -- Funda\c{c}\~ao para a Ci\^encia e a Tecnologia (Portugal) through PhD Grant SFRH/BD/151452/2021 (R.C.) and through CEECINST/00062/2018 (R.S.B.).
\end{acknowledgments}


\bibliography{references}


\end{document}